\documentclass[11pt]{article}
\usepackage[utf8]{inputenc}
\usepackage[margin=1in]{geometry}
\usepackage{authblk}

\title{Ultrafast Distributed Coloring of High Degree Graphs}

\author[1]{Magn\'us M. Halld\'orsson\thanks{Partially supported by Icelandic Research Fund grant 174484-051.}}
\newcommand\CoAuthorMark{\footnotemark[\arabic{footnote}]}
\author[1]{Alexandre Nolin\protect\CoAuthorMark}
\author[2]{Tigran Tonoyan\thanks{Partially supported by the European Union’s Horizon 2020 Research and Innovation Programme under grant agreement no. 755839.}}
\affil[1]{ICE-TCS \& Department of Computer Science, Reykjavik University, Iceland.}
\affil[2]{Technion  -- Israel Institute of Technology, Israel}

\usepackage{graphicx}
\usepackage{amsmath}
\usepackage{amsthm}
\usepackage{amssymb}
\usepackage{amsfonts}
\usepackage{algorithm}
\usepackage{algorithmic}
\usepackage{paralist}
\usepackage{color}
\usepackage{mathtools}
\usepackage{xspace}
\usepackage{hyperref}

\usepackage{thmtools}
\usepackage{thm-restate} 

\usepackage{xcolor}
\definecolor{mygray}{gray}{0.95}
\usepackage{mdframed}
\newenvironment{shadetheorem}
  {\begin{mdframed}[backgroundcolor=mygray,roundcorner=10pt,leftmargin=20,rightmargin=20,outerlinewidth=2,innertopmargin=0]\begin{theorem}}
  {\end{theorem}\end{mdframed}}

\newcommand{\parheader}[1]{\emph{#1}}

\makeatletter
\newcommand*{\bx}{bx}
\newcommand*{\IfBold}{
  \ifx\f@series\bx
    \expandafter\@firstoftwo
  \else
    \expandafter\@secondoftwo
  \fi
}
\makeatother

\newcommand{\alg}[2][]{{\IfBold{\MakeUppercase{#2}}{\textsc{#2}}}{#1}\xspace}
\newcommand{\trycolor}[1][]{\alg[#1]{TryColor}}
\newcommand{\tryrandomcolor}{\alg{TryRandomColor}}
\newcommand{\computeacd}{\alg{ComputeACD}}

\newcommand{\disjointsample}{\alg{DisjointSample}}
\newcommand{\slackgeneration}[1][]{\alg[#1]{GenerateSlack}}
\newcommand{\synchronizedcolortrial}{\alg{SynchColorTrial}}
\newcommand{\multitrial}[1][]{\alg[#1]{MultiTrial}}
\newcommand{\slackcolor}[1][]{\alg[#1]{SlackColor}}
\newcommand{\transversal}[1][]{\alg[#1]{Transversal}}
\newcommand{\lowdegreesample}[1][]{\alg[#1]{LowDegreeSample}}

\newcommand{\model}[1]{{\IfBold{\MakeUppercase{#1}}{\textsc{#1}}}\xspace}
\newcommand{\LOCAL}{\model{Local}}
\newcommand{\CONGEST}{\model{Congest}}
\newcommand{\eps}{\varepsilon}

\DeclareMathOperator{\poly}{poly}

\usepackage[capitalise]{cleveref}

\newtheorem{theorem}{Theorem}
\newtheorem{lemma}{Lemma}
\newtheorem{claim}{Claim}
\newtheorem{corollary}{Corollary}
\newtheorem{definition}{Definition}
\newtheorem{observation}{Observation}
\newtheorem{proposition}{Proposition}

\DeclareMathOperator*{\Exp}{\mathbb{E}}

\DeclarePairedDelimiter{\abs}{\lvert}{\rvert}
\DeclarePairedDelimiter{\card}{\lvert}{\rvert}

\DeclarePairedDelimiter{\set}{\lbrace}{\rbrace}
\DeclarePairedDelimiter{\event}{\lbrack}{\rbrack}
\DeclarePairedDelimiter{\range}{\lbrack}{\rbrack}
\DeclarePairedDelimiter{\parens}{\lparen}{\rparen}

\DeclarePairedDelimiter{\ceil}{\lceil}{\rceil}

\newcommand{\knuthupuparrow}{\mathbin{\uparrow\uparrow}}

\newcommand{\bbN}{\mathbb{N}}
\newcommand{\bbZ}{\mathbb{Z}}

\newcommand{\Bad}{\ensuremath{\mathtt{BAD}}} 
\newcommand{\colSpace}{\mathcal{C}} 
\newcommand{\pal}{\Psi} 
\newcommand{\col}{\psi} 

\newcommand{\acset}{\mathcal{S}_{\mathrm{ac}}} 

\newcommand{\core}{M} 

\newcommand{\cs}{\kappa} 

\newcommand{\pgen}{p_{\mathrm{g}}} 
\newcommand{\pdisj}{p_{\mathrm{s}}} 

\newcommand{\HFset}{\mathcal{H}} 
\newcommand{\hit}[4][]{#2\setminus^{#1}_{#3}#4} 
\newcommand{\samp}{\sigma} 
\newcommand{\out}{\lambda} 
\newcommand{\famsize}{F} 

\newcommand{\smin}{s_{\min}}
\newcommand{\sminpow}{\rho}
\newcommand{\iratio}{\vartheta} 

\newcommand{\Vsp}{V_{\mathrm{sparse}}} 
\newcommand{\csp}{c_{\mathrm{sp}}} 

\newcommand{\Det}{\textsc{Det}}
\newcommand{\Detd}{\textsc{Det}_{\mathrm{d}}}

\begin{document}

\maketitle

\begin{abstract}
    We give a new randomized distributed algorithm for the $\Delta+1$-list coloring problem.
The algorithm and its analysis dramatically simplify the previous best result known of Chang, Li, and Pettie [SICOMP 2020]. This allows for numerous refinements, and in particular, we can color all $n$-node graphs of maximum degree $\Delta \ge \log^{2+\Omega(1)} n$ in $O(\log^* n)$ rounds.
The algorithm works in the {\CONGEST} model, i.e., it uses only $O(\log n)$ bits per message for communication. 
On low-degree graphs, the algorithm shatters the graph into components of size $\operatorname{poly}(\log n)$ in $O(\log^* \Delta)$ rounds, showing that the randomized complexity of $\Delta+1$-list coloring in \CONGEST depends inherently on the deterministic complexity of related coloring problems.
\end{abstract}

\section{Introduction}

The graph coloring problem is one of the most fundamental to distributed computing. In fact, it was the topic of the first work on distributed algorithmics, by Linial in 1987 \cite{linial87}. This work also heralded the \LOCAL model, where nodes communicate with neighboring nodes in synchronous rounds, with no limit on computation nor message size. 
The default distributed coloring problem asks for a $\Delta+1$-coloring, where $\Delta$ is the maximum degree of the graph, while in the list variant, the nodes each have a palette of $\Delta+1$ admissible colors. 

There has been much progress on these problems, particularly in recent years. For $n$-node graphs, the deterministic complexity  is currently $O(\log^3 n)$ \cite{GK20} and the randomized complexity $O(\log^3 \log n)$ \cite{CLP20}. The similarity of these bounds is no coincidence. The randomized algorithm of Chang, Li and Pettie \cite{CLP18,CLP20} \emph{shatters} the graph in $O(\log^* n)$
rounds into connected components of poly-logarithmic size, on which it runs a deterministic algorithm for $\deg+1$-list-coloring (where the palette size of a node is its degree plus 1) \cite{GK20}.  
Chang, Kopelovitz and Pettie \cite{CKP19} showed that this actually runs both ways:
significant improvement in the randomized complexity must also imply
improved deterministic complexity.
This might lead one to mistakenly conclude that the \emph{tour de force} result of \cite{CLP20} was the end of story for randomized distributed graph coloring, with the remaining action only involving deterministic methods.

%
There are several reasons for desiring an improved
randomized coloring algorithm.
First of all, the intricacy of the algorithm and its analysis is an impediment to wide dissemination and applications. The algorithm features a hierarchy of $\log\log \Delta$ decompositions that are partitioned into ``blocks'', split by size, and combined into six different sets. These are whittled down in distinct ways, resulting in three final subgraphs that are finished off by two different deterministic algorithms.
The analysis of just one of these sets runs full 10 pages in the journal version \cite{CLP20}. 
This is unfortunate since the dominating performance of the method would make it an ideal candidate for extensions to related problems or transfer to other computational models.

The second reason why a different treatment is needed is that the result of \cite{CLP20} does not distinguish between the complexity of the problem for different values of $\Delta$. While the use of deterministic algorithms is necessary
for randomized coloring of graphs with $\Delta = O(\sqrt{\log n})$ \cite{CKP19}, it is not clear if it is needed for larger degrees.
Thirdly, the CLP algorithm uses large messages, of size polynomial in $\Delta$, which seems difficult to avoid as stated. And finally, we would like to be able to ask more refined questions about coloring, such as tradeoffs between round complexity and the number of colors, and reduced number of colors when certain structural properties hold.

\paragraph{Our results}
We give a coloring algorithm that drastically simplifies the previous best method known of \cite{CLP20}. The main step of \emph{reducing the dense components} is achieved in a \emph{single round} that avoids any major information gathering.
Overall, our algorithm works in the {\CONGEST} model, i.e., nodes communicate using small messages of size $O(\log n)$. We state our result succinctly: 

\begin{shadetheorem}\label{thm:intrologstar}
There is a randomized distributed algorithm for $\Delta+1$-list coloring $n$-node graphs of maximum degree $\Delta=\Omega(\log^{2+\Omega(1)} n)$ in $O(\log^* n)$ rounds of \CONGEST.
\end{shadetheorem}

In particular, it shows that \emph{shattering is only needed for low-degree graphs.} Previously, a better than $o(\log\log)$-time $\Delta+1$-coloring algorithm was only known for bounded $\Delta$. This suggests an unusual divide/dichotomy: graphs of either constant-degree \emph{or} degree $\log^{2+\Omega(1)} n$ have log-star complexity, while graphs in the middle appear significantly harder. 
    
For graphs of smaller degree, 
we improve slightly the recent $O(\log^5 \log n)$-round \CONGEST algorithm for $(\Delta+1$)-list-coloring \cite{HKMT21}.
More importantly, we show that for \CONGEST, like for \LOCAL, the randomized complexity is upper bounded by the logarithm of the deterministic complexity of the $\deg+1$-list coloring problem.
Effectively, all the advanced techniques that seemed to give \LOCAL more power when coloring graphs have now been achieved or co-opted in \CONGEST.

\begin{shadetheorem}\label{thm:introshattering}
There is a randomized distributed algorithm for $\Delta+1$-list coloring $n$-node graphs in $O(\log^3 \log n)$ rounds of \CONGEST.
In $O(\log^* \Delta)$ rounds, the algorithm shatters the graph into $\poly(\log n)$-size components.
\end{shadetheorem}

This means that if the deterministic complexity of $\deg+1$-list coloring and $\Delta+1$-list coloring problems is the same -- as previous algorithms might suggest -- then our algorithm is optimal. \footnote{This also presumes that the optimal such algorithms aren't significantly affected by the size of node IDs.}

We also obtain some refinements that follow easily from our approach.
If allowed $\Delta + O(\log^{1+\delta} n)$ colors, then 
$O(\log^* n)$-time suffices, for any fixed $\delta>0$ (\cref{cor:cor_of_main1}), answering a question posed by Chang, Li and Pettie \cite{CLP20}. 
The number of colors used by our algorithm is only 
$\Delta - \Omega(\Delta/\omega)$, when the clique number $\omega$ is significantly smaller than the maximum degree $\Delta$ (\cref{thm:cliquecol}).
And, finally, we obtain as a corollary a $O(\log^* \Delta)$-round \CONGEST algorithm for $2\Delta-1$-list edge coloring, when $\Delta \ge \log^{1+\Omega(1)} n$, extending a recent result for the non-list version \cite{representativesets}.

\paragraph{Our Techniques} We give a technical introduction in~\cref{sec:techintro} after presenting key technical definitions, concepts, and results. Here, we give a brief summary.

Our simplified coloring framework is based on an improved structural understanding of locally dense subgraphs. This allows us to color the bulk of these subgraphs with an exceptionally simple \emph{single-round} procedure, using only the colors \emph{suggested by a particular leader node}, which makes coordinated color choices within dense subgraphs easy, without the need of communicating palettes. In order to create slack colors for nodes with high neighborhood density, we ``put-aside'' part of their neighborhoods to be colored later. To this end, we introduce a technique based on \emph{independent transversals}.
The most technical aspect of the paper is coloring locally sparse nodes in the \CONGEST model. For that, we give a {\CONGEST} implementation of a {\LOCAL} procedure by \cite{SW10}, where a node may propose as many as $\log n$ colors to its neighbors. This is achieved by constructing a pseudorandom family of hash functions, termed \emph{representative hash functions}, that is sparse enough to allow communication in $O(\log n)$ bits. This general technique may be of wider interest.

\subsection{Related work}
The first paper to explicitly study the distributed coloring problem was a seminal paper by Linial~\cite{linial87} that effectively also started the area of distributed graph algorithms. He gave a deterministic algorithm that produces a $O(\Delta^2)$-coloring in $O(\log^* n)$ rounds, and showed that $\Omega(\log^* n)$ rounds were needed to color even the cycle graph $C_n$ with any constant number of colors.

Already in 1987, randomized parallel algorithms were known for $\Delta+1$-coloring that implied $O(\log n)$-round distributed algorithms.
Over the years, the randomized complexity of $\Delta+1$-coloring in \LOCAL improved to $O(\Delta \log\log n)$ \cite{KuhnW06}, $O(\log \Delta + \sqrt{\log n})$ \cite{SW10}, $O(\log \Delta + \Detd)$ \cite{BEPSv3}, $O(\sqrt{\log \Delta} + \log\log n + \Detd)$ \cite{HSS18}, and the current best $O(\Detd)$ \cite{CLP18}, where $\Detd$ denotes the deterministic complexity of $\deg+1$-list coloring on $\poly(\log n)$-vertex graphs.
In particular, Barenboim, Elkin, Pettie, and Su \cite{BEPSv3} showed that after applying coloring each node with a sufficiently high probability, the remaining graph is "shattered" into $\poly(\log n)$-sized components on which they apply a deterministic $\deg+1$-list coloring algorithm.
Chang, Kopelovitz, and Pettie \cite{CKP19} later proved that the randomized complexity is at least $\Det(\sqrt{\log n})$, where $\Det(n')$ denotes the deterministic complexity of $\Delta+1$-coloring on $n'$-node graphs. Under the assumption that $\Det$ and $\Detd$ are of similar magnitude, this would mean that the shattering approach is essentially optimal.

The deterministic complexity of coloring in \LOCAL has been tightly connected with the existence of \emph{network decompositions} (ND) \cite{awerbuch89,panconesi1992improved}. 
In a recent breakthrough, Rozho\v{n} and Ghaffari~\cite{RG19} showed that NDs can be constructed in $\poly(\log n)$ time, resulting in a $\poly(\log n)$-time algorithm for $\Delta+1$-coloring, later improved to $O(\log^5 n)$ \cite{GGR20}.
For small values of $\Delta$, the best deterministic complexity known is $O(\sqrt{\Delta\log \Delta} + \log^* n)$ \cite{FHK,Barenboim16,MT20}. 

Many of the early algorithms hold immediately in the \CONGEST model 
\cite{alon86,luby86,linial92,johansson99,KuhnW06,BarenboimEK14,BEG18}.
Some recent work deals explicitly  with $\Delta+1$-coloring in the \CONGEST model. Ghaffari \cite{Ghaffari2019} gave a randomized $O(\log\Delta) + \poly(\log\log n)$-round algorithm. 
Bamberger, Kuhn, and Maus \cite{BKM19} gave a deterministic $\poly(\log n)$-round algorithm, building on the improved NDs \cite{RG19}.
More recently, Ghaffari and Kuhn \cite{GK20} gave a $O(\log^2 \Delta \log n)$-round deterministic algorithm that does not use NDs, and works also in $\CONGEST$.
%
%
The first $\poly(\log\log n)$-round randomized algorithm for $\Delta+1$-coloring in \CONGEST was given earlier this year by Halld\'orsson, Kuhn, Maus, and Tonoyan \cite{HKMT21}, and we 
build on several of their results and insights.
While the overriding term $O(\log^5 \log n)$ of its round complexity is from the construction of the most efficient ND known \cite{GGR20} and might therefore be improved, it also has $\log^2 \log \Delta$ and $\log\log n$ terms that will not be reduced by improved deterministic algorithms. It also cannot leverage the more efficient deterministic algorithm of \cite{GK20}, due to the dependence of its complexity on the color space.

\section{Preliminaries}

\subsection{The Model and Basic Notation}

All our results are in the classic distributed computing models \LOCAL and \CONGEST,
where the nodes of a graph $G=(V,E)$ are computing agents of unlimited computational power that communicate only with neighboring nodes. The models are synchronous, and in each round, each node can send an individualized message to each of its neighbors. The models differ in that messages can be arbitrarily large in \LOCAL, but are restricted to $O(\log n)$ bits in \CONGEST, where $n=|V|$ is the number of nodes. Each node is also assumed to have a private source of randomness, and we assume they all know the maximum degree $\Delta$ of $G$ and a (common) polynomial bound on $n$. 

Every node $v$ has a palette of available colors, denoted by $\pal(v)$, given before the start of the algorithm.
In the coloring problems we consider, each node should assign itself a color from its palette different from its neighbors. 
For $\Delta+1$-coloring, $\pal(v)=\set{1,\ldots,\Delta+1}$, for all $v$; in \emph{$\Delta+1$-list} coloring, $\pal(v)$ is an arbitrary $\Delta+1$-sized list of colors from some commonly known color space~$\colSpace$; while in $\deg+1$-list coloring, $\pal(v)$ is of size $d_v+1$, where $d_v$ is the degree of $v$.

For a node $v$, let $N(v)$ denote its set of neighbors and note that $d_v = \card{N(v)}$.
For a set $S \subseteq V$, let $N_S(v) = N(v) \cap S$.
Let $m(S)$ denote the number of edges with both endpoints in set $S\subseteq V$ of vertices. Throughout, $\eps$ will denote a small fixed positive constant known by all nodes.

\textbf{Conventions.} We will universally ignore or remove from the graph the nodes that get colored. In particular, any set $X$ of nodes is assumed to be dynamic: at any time, it contains the portion of its initial nodes that have not been  colored yet. Similarly, the palette $\pal(v)$ of a node contains only the colors of the initial palette that have not been permanently assigned to a neighbor of $v$ yet.

\subsection{Slack, Sparsity, \& Almost-Cliques}

We introduce in this section the main technical concepts behind sublogarithmic randomized coloring algorithms and then give a technical introduction to our extensions.

The basic primitive in randomized coloring algorithms, which we call {\tryrandomcolor}, is for nodes to \emph{try} a random eligible color: propose it to its neighbors and keep it if it doesn't conflict with them. More formally, we run {\trycolor} (\cref{alg:trycolor}), with an independently and uniformly sampled color $\col_v\in \pal(v)$.
Repeating it leads to a simple $O(\log n)$-round algorithm \cite{johansson99}.
\begin{algorithm}[H]\caption{\trycolor (vertex $v$, color $\col_v$)}\label{alg:trycolor}
\begin{algorithmic}[1]
\STATE Send $\col_v$ to $N(v)$, receive the set $T=\{\col_u : u\in N(v)\}$.
\STATE{\textbf{if}} $\col_v\notin T$ \textbf{then} permanently color $v$ with $\col_v$.
\STATE Send/receive permanent colors, and remove the received ones from $\pal(v)$.
\end{algorithmic}
\end{algorithm}

Having more colors to choose from --  quantified in the definition below --  makes the task easier.

\begin{definition}[Slack]
\label{def:slack}
The \emph{slack} of a node $v$ in a given round is the difference $|\pal(v)|-d$ between the number of colors it has then available and the number $d$ of neighbors of $v$ competing for these colors in that round.
\end{definition}
Initially, $v$ has slack $\Delta+1-d_v$. It can increase \emph{permanently}, both when two neighbors take the same color and when a neighbor takes a color outside $v$'s (original) palette. It can also increase \emph{temporarily}, while a set of neighbors sits out rounds (i.e., does not compete with $v$).

Schneider and Wattenhofer \cite{SW10} showed that coloring can be achieved ultrafast if all nodes have slack at least proportional to their degree (and the degree is large enough). This is achieved by each node trying up to $\log n$ colors in a round, using the high bandwidth of the \LOCAL model. 
One key technical contribution of our work is achieving this in \CONGEST (proved in Sec.~\ref{sec:congest}).

\begin{restatable}{lemma}{slackcolorlemma}
\label{lem:slackcolor}
Consider the $\deg+1$-list coloring problem where each node $v$ has slack $s_v=\Omega(d_v)$.  
Let $1<\smin \leq \min_v s_v$ be globally known.
For every $\delta\in (1/\smin,1]$, there is a randomized  \CONGEST algorithm {\slackcolor[$(\smin)$]} that in $O(\log^* \smin+1/\delta)$ rounds properly colors each node $v$ w.p.\ $1 - \exp(-\Omega(\smin^{1/(1+\delta)})) - n^{-\Theta(1)} - \Delta e^{-\Omega(\smin)}$, even conditioned on arbitrary random choices at distance $\geq 2$ from~$v$.
\end{restatable}

The \emph{sparsity} of a node measures how far its neighborhood is from a $\Delta$-clique. 

\begin{definition}[Sparsity]
\label{def:sparsity}
The \emph{(local) sparsity} $\zeta_v$ of node $v$ is defined as $\frac{1}{\Delta}\cdot\left[\binom{\Delta}{2}-m(N(v))\right]$. Node $v$ is \emph{$\zeta$-sparse} if $\zeta_v\ge \zeta$.
\end{definition}
 For any two nodes $u,v\in V$, $m(N(u))\ge m(N(u)\cap N(v))\ge m(N(v))-|N(v)\setminus N(u)|\Delta$, hence:
\begin{observation}\label{obs:sparsityrelation}
For any $u,v\in V$, $\zeta_u\le \zeta_v + |N(v)\setminus N(u)|$.
\end{observation}

Sparsity  
yields slack, by executing initially
the following single-round color trial.
The conversion of sparsity into slack was first shown by Reed \cite{Reed98,molloy2013coloring} and later popularized in distributed coloring problems by Elkin, Pettie and Su \cite{EPS15}.
We use here a  variant from \cite[Lemma 6.1]{HKMT21}.

\begin{lemma}[\cite{HKMT21}]
\label{L:sparseGetsSlack}
After {\slackgeneration}, each node $v$ gets slack $\Omega(\zeta_v)$, w.p.\ $1-\exp(-\Omega(\zeta_v))$. \end{lemma}

\begin{algorithm}[H]\caption{\slackgeneration}\label{alg:slackgeneration}
\begin{algorithmic}[1]
\STATE $S\gets $ sample each $v\in G$ into $S$ independently w.p.\ $\pgen= 1/20$.
\STATE \algorithmicforall\ $v\in S$ in parallel \algorithmicdo\ {\tryrandomcolor}$(v)$.
\end{algorithmic}
\end{algorithm}


Highly sparse nodes can be colored easily by {\slackcolor} (\cref{lem:slackcolor}) after {\slackgeneration}, as first argued in \cite{EPS15}, who used this for fast  $2\Delta-1$-edge coloring.
This shifts the focus to the \emph{dense} nodes (of sparsity $o(\Delta)$). 

Harris, Schneider and Su \cite{HSS18} proposed the following decomposition of a graph into a sparse part and a collection of dense subgraphs. It is central to all known superfast coloring algorithms \cite{HSS18,CLP18,CLP20,HKMT21}. 
We use an extension given by Assadi, Chen and Khanna \cite{ACK19}.
\begin{definition}[ACD, \cite{HSS18,ACK19}] \label{def:acd}
Let  $G=(V,E)$ be a graph with maximum degree $\Delta$, and $\eps\in (0,\frac{1}{3})$. A partition $V=\Vsp \cup \bigcup_{C \in \acset} C$ of $V$ is an \emph{$\eps$-almost-clique decomposition (ACD)} for $G$ if: 
\begin{compactenum}
    \item $\Vsp$ consists of $\Omega(\eps^2\Delta)$-sparse nodes\ ,
    \item For every  $C \in \acset$, $\card{C}\le (1+\eps)\Delta$\ ,
    \item For every $C \in \acset$ and $v\in C$,   $|N_{C}(v)|\ge (1-\eps)\Delta$\ .
\end{compactenum}
\end{definition}
An ACD can be found in $O(1)$ rounds, both in {\LOCAL} \cite{HSS18} and {\CONGEST} \cite{HKMT21}. 
We refer to the $C$'s as \emph{almost-cliques}. It follows from properties 2 and 3 that 
the diameter of each $G[C]$ is at most 2, opening the possibility of synchronizing the actions of the nodes in $C$. 

The basic primitive of the previous algorithms \cite{HSS18,CLP18,CLP20,HKMT21} for dealing with  dense nodes  
is to synchronize the random color tries of the nodes in $C$ so that they don't conflict with each other. A color tried then conflicts only with those selected by the node's \emph{external neighbors} outside $C$.

\begin{definition}[External/anti-degree]
For a node $v \in V \setminus \Vsp$, let $C_v$ denote its almost-clique,
$E_v = N(v)\setminus C_v$ its set of \emph{external neighbors} and $e_v = |E_v|$ its \emph{external degree}.
Similarly, let $A_v = C_v \setminus N(v)$ denote its set of \emph{anti-neighbors} and $a_v = |A_v|$ its \emph{anti-degree}.
\end{definition}

Note that our definition of external neighbors differs from that of previous works \cite{HSS18,HKMT21} 
in that it includes neighbors in $\Vsp$ in addition to neighbors in other almost-cliques. With the narrower definition, it was recently observed \cite{HKMT21} that the external degree of a node is actually bounded by its sparsity, and thus (probabilistically) by its slack.
We extend this result to our more inclusive definition of external degree.

\begin{restatable}{lemma}{acdproperties}
\label{lem:acdproperties}
Assume $\eps<1/3$, let $C$ be an almost-clique, and let $\csp \le 1$ be a constant such that all nodes $u\in \Vsp$ have sparsity $\zeta_u \geq \csp \eps^2 \Delta$. For every $v\in C$, $e_v\le 4\zeta_v/(\csp\eps^2)$, $a_v\le 2\zeta_v/(1-3\eps)$.
\end{restatable}

\begin{proof}
To bound $e_v$, we may assume $\zeta_v < (1/2)\csp \eps^2 \Delta$, as otherwise $e_v\leq \eps \Delta \le 2\zeta_v/(\csp \eps)$.  Let us count the edges in $v$'s neighborhood. Letting $n_{v,u}=\card*{N(v)\setminus N(u)}$, we have $\binom{\Delta}{2} - \zeta_v\Delta=m(N(v))  = \frac 1 2 \sum_{u \in N(v)} \card{N(u) \cap N(v)} \leq \frac 1 2 \sum_{u \in N(v)} \parens*{\Delta-1-n_{v,u}} \leq \binom{\Delta}{2} - \frac 1 2 \sum_{u \in N(v)} n_{v,u}$. Rearranging terms, this gives $\zeta_v \ge \frac{1}{2\Delta} \sum_{u \in N(v)} n_{v,u}$.
    Let us now consider an external neighbor $u\in E_v$. If $u \in C' \neq C$, then $n_{v,u}\geq (1-2\eps)\Delta$. If $u \in \Vsp$, we have $\csp \eps^2 \Delta\le \zeta_u\le \zeta_v + n_{v,u}$ (by \cref{obs:sparsityrelation}), hence  $n_{v,u}\ge \zeta_u-\zeta_v\ge (1/2)\csp\eps^2\Delta$.
    Therefore, an external neighbor of either type contributes at least $(1/4)\csp\eps^2$ to $v$'s sparsity, i.e., $e_v \leq 4\zeta_v/(\csp \eps^2)$.

The bound on $a_v$ is from \cite[Lemma 6.2]{HKMT21}. We give a proof in \cref{app:missing-proofs} for completeness. 
\end{proof}

This crucially means that external neighbors are not hurdles \emph{per se} for applying {\slackcolor} (\cref{lem:slackcolor}). This is a key distinction from \cite{HSS16} and \cite{CLP18,CLP20} that used a different sparsity definition which is only quadratically related to slack. 

What remains is to reduce the \emph{internal degree} of each node, or the size of its almost-clique, but the diversity of the external degrees can be a challenge.

\subsection{Technical Introduction}
\label{sec:techintro}

\parheader{Separating outliers.} 
We first show that the majority of nodes in an almost-clique $C$ have the same sparsity $\zeta_C$ (up to a constant factor), which allows us to analyze them ``in bulk''.
For this, we first extract a constant-fraction subset $O$, \emph{the outliers}, of higher sparsity. They can be disposed of first by {\slackcolor}, since with $C\setminus O$ inactive, they have (temporary) slack $\Omega(\Delta)$.

\parheader{Single synchronized color trial.}
Before applying {\slackcolor} on $C \setminus O$, we need to reduce the \emph{internal degree} of its nodes, i.e., the \emph{size} of $C$. This is the key operation in all known superfast coloring algorithms \cite{HSS18,CLP20,HKMT21}. We find that this can actually be achieved in a \emph{single round} of coordinated color trials. The intuition is that if the nodes in $C$ coordinate their trials, they only conflict with their external neighbors, so they have failure probability $O(e_v/\card{C}) = O(\zeta_C/\Delta)$, which leaves $O(\zeta_C)$ nodes remaining in $C$. Since the nodes also have slack $\Omega(\zeta_C)$, 
we can color all the nodes of $C$ with {\slackcolor} in $O(\log^* n)$ rounds, as long as $C$ has sparsity $\zeta_C = \Omega(\log^{1+\Omega(1)} n)$.

In fact, this key step is extremely simple: a pre-elected leader $w_C$ in $C$ gives distinct random colors from its palette to all nodes in $C$ to try. This obviates any communication involving palettes or topology. In the non-list variant, one can take the leader to be the node with the most neighbors within $C$. Because of the bounds on external- and anti-degrees, this ensures that the palette of $w_C$ and that of any other node differs in only $O(\zeta_C)$ colors. 

\parheader{Centrality of palettes.}
For the list-variant, we show that it is enough to pick as leader the node with the \emph{most typical palette}. This is found by relating (probabilistically) the discrepancy of palettes within an almost-clique to the \emph{chromatic slack}, i.e., the type of slack generated when a neighbor assumes a color not within a node's palette. The node with the smallest chromatic slack is then close to having the most central, or typical, palette in $C$.

\parheader{Slack for dense nodes via independent transversals.}
Nodes with very low sparsity (less than $\log^{1+\Omega(1)} n$) do not have enough slack to be fully colored with high probability by {\slackcolor}. Our solution is
to identify a \emph{put-aside} subset of each almost-clique that is large enough to provide slack (by inactivity) to the rest, while being easy to color in the end. These subsets, known as independent transversals, contain a significant fraction of each high-density almost-clique, while there are no edges between the subsets in different almost-cliques (which makes it easy to color them locally).
We use the fact that nodes in high-density almost-cliques have few external neighbors, so we can use a \emph{sample-and-correct} approach to find large such transversals, as long as $\Delta = \log^{2+\Omega(1)} n$.

\parheader{Multiple color trials in \CONGEST}.
Nearly all aspects of our algorithm are immediately implementable in \CONGEST, i.e., do not require large messages. The obstacle is the multi-trial method {\slackcolor} of Lemma \ref{lem:slackcolor}, where we want to send up to $\log n$ colors to a neighbor, while the bandwidth allows only a single color. 
A major technical contribution of this paper is to show how certain communication tasks like this -- involving the identification of $\Theta(\log n)$ different "values" along each edge -- can be achieved with $O(\log n)$ bits, and thus in a constant number of \CONGEST rounds. The idea is to implement an \emph{approximate} variant, with slightly weaker (but sufficiently strong) probabilistic guarantees.  This can be achieved via small carefully designed pseudorandom families of hash functions we call \emph{representative hash functions}: each node $v$ hashes its palette into $[\out_v]$, for $\out_v=\Theta(|\pal(v)|)$, and then picks its random color proposals from those hashed in $[\samp]$, for $\samp=O(\log n)$. With representative hash functions, this behaves for our purposes as if the color choices were independently random, and moreover, these choices can be efficiently communicated.

\smallskip

\textbf{Organization of the paper.} We present the new coloring algorithm in the upcoming Sec.~\ref{sec:main}, along with most of the analysis for non-list coloring. The list coloring modifications are given in Sec.~\ref{sec:listcol}, and the \CONGEST implementation of {\slackcolor} is in Sec.~\ref{sec:congest}. Some refinements and improvements are deferred to Sec.~\ref{sec:refinements}. Appendix~\ref{app:concentration} lists the concentration bounds used, while proofs of some known results are in Appendix~\ref{app:missing-proofs}.

\section{Improved Randomized Coloring Algorithm} 
\label{sec:main}

We present a new randomized coloring algorithm (\cref{alg:logstar}) that is simpler, works faster on high-degree graphs, and requires only small messages.

\begin{algorithm}[H]
\caption{Randomized $\Delta+1$-Coloring Algorithm}
\label{alg:logstar}
  \begin{algorithmic}[1]
  \STATE {\computeacd}.
    \STATE {\slackgeneration}.
    \STATE Each almost-clique $C$ computes its leader $w_{C}$ and outliers $O_C$. \label{st:outliers}
    \STATE {\slackcolor} in $G[\Vsp \cup O]$. \label{st:o-multitrial}
    \STATE $P_C \gets$ {\disjointsample}($C$, $\Delta^{1/3}$) for almost-cliques $C$ with $\zeta_{C} \le \Delta^{1/3}$.     \label{st:putaside}
    \STATE {\synchronizedcolortrial} in $G \setminus P$.\label{st:synchtrial}

    \STATE {\slackcolor} in $G\setminus P$. \label{st:lastmultitrial}
    \STATE For each $C$, let $w_C$ collect the palettes in $P_C$ and color the nodes locally.\label{st:collect}
\end{algorithmic}
\end{algorithm}

We start with the generation of ACD (using the algorithm from \cite{HKMT21}) and slack. The novelty is in coloring the almost-cliques, which are split into three parts: \emph{outlier} nodes $O$ of overly high sparsity (line \ref{st:outliers}); \emph{put-aside} set $P$ formed by disjoint portions of the almost-cliques (line \ref{st:putaside}); and the remaining main part, $C \setminus (O \cup P)$.
The outliers are colored first by {\slackcolor} using $C\setminus O$ for slack (line \ref{st:o-multitrial}); 
next the main part goes through synchronized color trials (line \ref{st:synchtrial}) followed by {\slackcolor} (line \ref{st:lastmultitrial}), using both $P$ and the uniform sparsity of the nodes for slack; and finally, the components of $P$ are colored locally without slack (line \ref{st:collect}).

We treat the different subroutines (lines \ref{st:outliers}, \ref{st:putaside}, and \ref{st:synchtrial}) in the following lemmas, before arguing the complexity of the algorithm.

\begin{figure}[t!]
    \centering
    \includegraphics[height=0.225\linewidth,page=4,trim=0em 0em 8.5em 0em,clip]{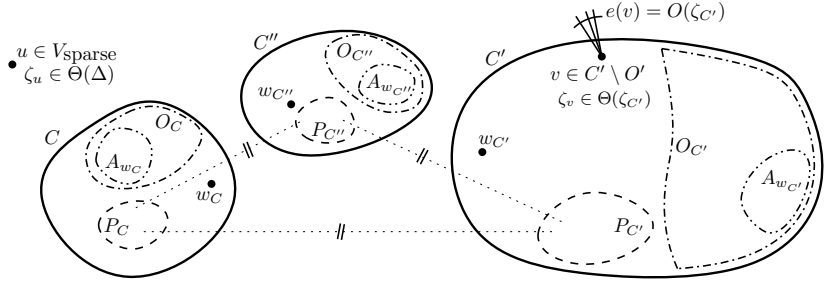}
    \caption{\small{Illustration of our partitioning. Most importantly, no edges connect $P_{C}$ and $P_{C'}$ for $C\neq C'$; $\card{C\setminus O_{C}}=\Theta(\Delta)$; $\zeta_v = \Theta(\zeta_{C})$ for $v \in C \setminus O_{C} \subseteq N(w_{C})$; and $\zeta_v \in \Omega(\Delta)$ for $v\in \Vsp$.}}
    \label{fig:our_partitionning}
\end{figure}

\paragraph{The leader and the outliers}
The \emph{leader} $w_C$ of an almost-clique $C$ is a node of smallest \emph{anti-degree} $a_{w_C}$ in $C$. Let $\zeta_C = \zeta_{w_C}$ be its sparsity. 
For each almost-clique $C$, we separate a subset 
\[
O_{C}=A_{w_C}\cup \{u\in N_C(w_C) : |N(u)\cap N(w_C)|< \Delta-5\zeta_{C}\}
\]
of \emph{outlier} nodes that may have significantly larger sparsity than $\zeta_C$.  
We also let $\core_C=C\setminus O_C$ denote the \emph{main} part of $C$, and  $O=\cup_{\acset} O_{C}$ the set of all outliers. 

The following observation, a key to our results, shows that after eliminating a modest fraction of each almost-clique, the remaining nodes have roughly the same sparsity. Since sparsity is the link between external degree and slack, it suffices to reduce the nodes' internal degrees (by {\synchronizedcolortrial}) 
in order to apply \slackcolor.

\begin{lemma}\label{lem:removeo}
For every almost-clique $C$ and node $v\in \core_C$, 
$\parens*{\frac 4 {\csp \eps^2}+\frac {3-3\eps}{1-3\eps}}^{-1}\zeta_C\le \zeta_v\le 6\zeta_C$.
Further, after {\slackgeneration}, it holds for each  $w\in C$ w.p.\ $1-e^{-\Omega(\Delta)}$ that $|N_{\core_C}(w)|\ge (1/2-2\eps)\Delta$.
\label{L:osmall}
\end{lemma}
\begin{proof}
The first inequality of the first claim follows by using \cref{obs:sparsityrelation}, the definition of $w_C$, and \cref{lem:acdproperties}: for any $u\in C$, $\zeta_C\le \zeta_u+|N(u)\setminus N(w_C)|\le \zeta_u+a_{w_C}+e_u\le \zeta_u+a_u+e_u \le \parens*{\frac 4 {\csp \eps^2}+\frac {3-3\eps}{1-3\eps}}\zeta_u$.

For the second inequality, let $u \in \core_C$. 
By sparsity of $w_C$, there are at most $\zeta_C \Delta$ non-edges in $G[N(w_C) \cap N(u)] \subseteq G[N(w_C)]$. 
By definition of $O_C$, $|N(u)\cap N(w_C)|\ge \Delta-5\zeta_C$, so $N(u)$ contains at most $5\zeta_C$ nodes outside $N(w_C)$, which contribute at most $5\zeta_C \Delta$ to $\overline m_u=\binom{\Delta}{2}-m(N(u))$. Also adding the mentioned $\zeta_C\Delta$ non-edges, we get $\overline m_u\le 6\zeta_C \Delta$, i.e., $\zeta_u \le 6\zeta_C$. The first claim is proven.

     For the second claim, first note that by Def.~\ref{def:acd}, 
     $|N_{C}(w)|\ge (1-\eps)\Delta$  
     holds for each $w\in C$, before slack generation.
    In expectation, at most $\pgen\Delta= \Delta/20$ neighbors are colored by slack generation, and by Chernoff bound (Lemma~\ref{lem:basicchernoff}), at most $\Delta/10$ are colored, w.p.\ $1 - e^{-\Omega(\Delta)}$. 
To bound $O_C$, let us bound $m(N(w_C))$ in terms of the edges contributed by various nodes in $N(w_C)$.
Note that each node in $O_C\setminus A_{w_C}$ contributes at most $\Delta-5\zeta_C-1$ edges, while the rest of the nodes contribute at most $\Delta-1$. Letting $o_C=|O_C\setminus A_{w_C}|$, we get
$
m(N(w_C))\le (1/2)(o_C\cdot (\Delta-5\zeta_C-1) + (\Delta-o_C)\cdot (\Delta-1))=\binom{\Delta}{2}-(5/2)\zeta_C o_C
$; hence, $o_C\le 2\Delta/5$, and $|O_C|\le o_C+a_{w_C}\le (2/5+\eps)\Delta$.
Putting all together we see that  $|N_{\core_C}(w)|\ge (1-\eps-1/10-2/5-\eps)\Delta=(1/2-2\eps)\Delta$ holds with probability $1-e^{-\Omega(\Delta)}$. 
\end{proof}

\paragraph{Generating Slack for Dense Nodes}
Recall that slack generation (\Cref{L:sparseGetsSlack}) yields slack w.h.p.\ only for sufficiently sparse nodes. Denser nodes need additional slack in order to apply {\slackcolor}.
The solution is to provide them with temporary slack by identifying a subgraph that can be put aside and colored later in isolation. 
This is done with the deletion method in Alg.~\ref{alg:lds}: randomly sample the nodes of each highly dense almost-clique and only keep those that have no sampled neighbor outside their clique.

\begin{algorithm}[H]\caption{{\disjointsample}$(C, B)$} 
\label{alg:lds}
  \begin{algorithmic}[1]
  \STATE $S_C\gets$ each node $v\in \core_C$ is sampled independently w.p.\ $\pdisj=1/(cB)$
  \RETURN $P_C \gets \{v\in S_C : E_v\cap S = \emptyset\}$, where $S = \cup_{C'} S_{C'}$
 \end{algorithmic}
\end{algorithm}

\begin{lemma}\label{lem:lds}
Let $B = \Delta^{1/3}$ and $\pdisj=1/(cB)$, for a large enough constant $c>0$.
Suppose {\disjointsample}$(C,B)$ is run in all almost-cliques $C$
with $\zeta_{C} \le B$, returning a set $P_C$. 
Then, 
for each such $C$, $|P_C| = \Omega(B^2)$, 
w.p.\ $1 - \exp(-\Omega(B))$.
\end{lemma}

\begin{proof}
    As a first step, let us observe that by   Chernoff bound (\Cref{lem:basicchernoff} with  $q_i=\pdisj=1/(c B)$) and \cref{lem:removeo},
    $|S_C|\ge \Delta/(4cB)=B^2/(4c)$, w.p.\ $1-\exp(-\Omega(B^2))$. For a node $v\in S_C$, let $X_v$ be the indicator random variable that is 1 when an external neighbor of $v$ in an almost-clique $C'$ with $\zeta_{C'}\le B$ is sampled. Note that $P_C=\{v\in S_C : X_v=0\}$. Since each node $w$ is sampled w.p.\ $1/(cB)$ and $v$ has external degree $e_v= O(\zeta_C) \le c_1 B$ (\Cref{lem:acdproperties,lem:removeo}), we have $\Pr[X_v=1]\le c_1/c$, where $c_1$ depends only on $\eps$. Note that each variable $X_v$ is a function of the independent 
    indicator 
    variables $\{Y_w\}_{w\in E_v}$,  
    of the events that an external neighbor $w$  is sampled. 
    Since every node $w\in C'$, with $C'\neq C$ and $\zeta_{C'}\le B$, has at most $e_w\le c_1 B$ neighbors in $C$,  
    we see that for a given $S_C$,
    $\{X_v\}_{v\in S_C}$ is a read-$c_1 B$ family of random variables, and Lemma~\ref{lem:kread} applies (with $q\leq c_1/c$, $k=c_1 B$, $\delta=c_1/c$), showing that $|S_C|-|P_C|=\sum_{v\in S_C}X_v>  (2c_1/c)|S_C|$ holds w.p.\ less than $2\exp(-(2c_1/c^2)|S_C|/B))$. Thus, the probability that either $|S_C|<B^2/(4c)$ or $|P_C|<(1-2c_1/c)|S_C|$ is $\exp(-\Omega(B))$. 
    We let $c= 4c_1$, so that  
    $|P_C|\ge |S_C|/2\ge B^2/(32c_1)$, w.p.\ $1-\exp(-\Omega(B))$.
\end{proof}

\paragraph{Internal Degree Reduction}
Synchronizing color trials in dense components is fundamental to all known sublogarithmic-time $\Delta+1$-coloring algorithms. In \cite{HSS18}, such a primitive was applied $\sqrt{\log n}$ times; in \cite{HKMT21}, $\log^2 \log \Delta$ times; while 
in \cite{CLP18, CLP20}, two such primitives were defined and applied $O(1)$ times in different ways on  different subgraphs. 
Here we apply only once a particularly na\"ive such primitive that avoids any communication about the topology or the node palettes.

In {\synchronizedcolortrial}, the leader sends a random unused candidate color from \emph{its own palette} to each node in $\core_C$ (to try).  Note that even in {\CONGEST} this is easily done, and every node receives a color, since by the definition of $O$, $\core_C\subseteq N(w_C)$. 

\begin{algorithm}[H]\caption{\synchronizedcolortrial, for almost-clique $C$}
\label{alg:synchtrial}
  \begin{algorithmic}[1]
    \STATE $w_C$ randomly permutes its palette $\pal(w_C)$, sends each neighbor $u \in \core_C$ a distinct color $\col_u$. \label{st:randomorder}
    \STATE Each $u \in \core_C$ calls {\trycolor}($u$, $\col_u$) if $\col_u \in \pal(u)$
    \end{algorithmic}
\end{algorithm}

We bound how many nodes in $C$ fail to get colored by {\synchronizedcolortrial}, either because the candidate color they receive is outside their palette or because their color trial failed.

\begin{observation}\label{obs:leaderpalettenonlist}
For every $u\in C$, $|\pal(w_C)\setminus \pal(u)| \le 
\parens*{\frac{24}{\csp\eps^2} + \frac{2}{1-3\eps}}\zeta_C$, 
even when part of the graph is colored.
\end{observation}
\begin{proof}
A color that $w_C$ assigns to $u$ is not in $\pal(u)$ only if it has been taken by an external neighbor of $u$ or by a node in $N(u)\cap A_{w_C}$. At any time, there are at most $e_u+a_{w_C}$ such nodes, where 
$e_u\le 4\zeta_u/(\csp\eps^2)\le 24\zeta_C/(\csp\eps^2)$
and $a_{w_C}\le 2\zeta_C/(1-3\eps)$, by Lemmas~\ref{lem:acdproperties} and~\ref{lem:removeo}.
\end{proof}

A node is \emph{decolored} if it gets a candidate color not in its palette or if it tries a color but fails to keep it.
The following key lemma shows that a single {\synchronizedcolortrial} suffices to reduce the size of an almost-clique to its sparsity, paving the way for the application of {\slackcolor}.

\begin{lemma}\label{lem:degred} Let $C$ be an almost-clique, and let $t=\Omega(\zeta_C)$. W.p.\ $1-\exp(-t)$, the number of decolored nodes of $C$ in step~\ref{st:synchtrial} in \cref{alg:logstar}, is $O(t)$. 
\end{lemma}

\begin{proof}
Fix arbitrary candidate colors for nodes outside $C$ -- we prove the success of the algorithm within $C$ for arbitrary behaviors outside $C$.
Let $C'$ be an arbitrary subset of $\core_C$ of size $\ceil{\card{\core_C}/2}$. 
    For each $u \in C'$, recall that $\col_u \in \pal(w_C)$ is its candidate color, and let $X_u$ be the binary r.v.\ that is 1 iff $u$ is decolored. 
    Consider a node $u\in C'$. Conditioning on an arbitrary set $S$ of $|C'|-1$ candidate colors assigned to the nodes in $C'\setminus u$, $\col_u$ is uniformly distributed in $\pal(w_C)\setminus S$, which has size $|\pal(w_C)|-|C'|+1\ge  (1/4-\eps)\Delta$ (using \cref{lem:removeo}). The node $u$ is decolored only when its candidate color is also tried by one of its external neighbors, of which it has $|E_{u}|=O(\zeta_C)$ (Lemmas~\ref{lem:acdproperties} and~\ref{lem:removeo}), or when it is not in its palette, i.e., when it belongs to $\pal({w_C})\setminus \pal(u)$. Note that $|\pal({w_C})\setminus \pal(u)|=O(\zeta_C)$, by Observation~\ref{obs:leaderpalettenonlist}. Thus, $\Pr\event{X_u=1\mid \{\col_w\}_{w\ne u}}\le q$, for $q=O(\zeta_C/\Delta)$.
    Having fixed the candidate colors of nodes outside $C$, each $X_u$ is determined by $\col_u$, so we also have $\Pr\event{X_u=1\mid \{X_w\}_{w \ne u}}\le q$.
    Applying \cref{lem:chernoff}, we get that
    $\Pr\event{\sum_{u \in C'} X_u\le 4t}\ge 1-\exp(-t)$ for any $t \ge q |C'| = \Omega(\zeta_C)$. By symmetry, the same holds for $\core_C \setminus C'$, and the lemma follows by the union bound.
\end{proof}

\paragraph{The Main Result}
\label{sec:result2}
We are ready to argue a simplified version of our main result (\cref{thm:intrologstar}).
We extend it in Sec.~\ref{sec:listcol} to $\Delta+1$-list coloring and in 
Sec.~\ref{sec:betterhideg} to graphs with $\Delta=\log^{2+\Omega(1)} n$.

\begin{theorem}
There is a randomized \CONGEST $\Delta+1$-coloring algorithm with runtime $O(\log^* n)$, for graphs with $\Delta =\Omega( \log^{3+\delta} n)$, for any constant $\delta>0$. 
\label{thm:log-star}
\end{theorem}

\begin{proof}
%
The nodes of $\Vsp \cup O$ have slack $\Omega(\Delta)$ in $G[\Vsp\cup O]$ (by Lemma~\ref{lem:removeo}, since each node in $O$ has $\Omega(\Delta)$ uncolored neighbors outside of $\Vsp\cup O$), and are then fully colored by {\slackcolor} in $O(\log^* n)$ rounds, w.h.p.\ (Lemma~\ref{lem:slackcolor}).

Consider an almost-clique $C$. 
Every node $v\in \core_C\setminus P_C$ has slack $s_C = \Omega(\zeta_C+\Delta^{1/3}) = \Omega(\log^{1+\delta/3} n)$, w.h.p., either from {\slackgeneration} (Lemma \ref{L:sparseGetsSlack}), when $\zeta_C = \Omega(\Delta^{1/3})$, or from the at least $|P_C|-a_v = \Omega(|P_C|)$ neighbors in $P_C$ (Lemma \ref{lem:lds}), when $\zeta_C=O(\Delta^{1/3})$. 
The nodes of $C \setminus P_C$ have degree $O(\zeta_C+\Delta^{1/3}) = O(s_C)$ following  {\synchronizedcolortrial}: $O(\zeta_C)$ external neighbors (by \cref{lem:acdproperties}),  
and $O(\zeta_C+\Delta^{1/3})$ decolored internal neighbors (\cref{lem:degred}, with $t=\zeta_C+\Delta^{1/3}$).
Hence, all nodes in $\core_C\setminus P_C$ will be colored in $O(\log^* n)$ rounds by {\slackcolor} (applied with  $\smin=\Omega(\Delta^{1/3})$ and $\delta'=\delta/3$), w.h.p.

Finally, the put-aside sets $P_C$ and $P_{C'}$, for any $C\ne C'$, have by construction no edge between them, and $P_C\subseteq N_C(w_C)$, so $P_C$ can be colored locally within $C$. We explain at the end of this section how to do that in $O(1)$ rounds in {\CONGEST}. 
 \end{proof}

For small values of $\Delta$, we still run our algorithm but instead of obtaining high probability bounds, we move nodes that fail probabilistic properties into a set $\Bad$. 
The nodes in $V\setminus \Bad$ are fully colored, and the set $\Bad$ is colored at the very end. We show that 
the subgraph $G[\Bad]$ is \emph{shattered}, in that 
it consists of $\poly(\log n)$-sized connected components, w.h.p.\ (cf. \cref{thm:introshattering}).

\begin{theorem}
There is a $O(\log^* \Delta)$-round \CONGEST algorithm for \emph{shattering}.
The randomized complexity of $\Delta+1$-coloring is at most the deterministic complexity of $\deg+1$-list coloring on instances of size $\poly(\log n)$. In particular, it is $O(\log^3 \log n)$.
\label{thm:shattering}
\end{theorem}
Note that the deterministic algorithms need to be able to handle large node IDs. Namely, they are run on instances of size $n' = \poly(\log n)$, where  nodes have $\Theta(\log n) = \poly(n')$-bit IDs. 

The advantage of our shattering method for \CONGEST over previous methods is that since we only resort to it when $\Delta = \poly(\log n)$, the color values are only $O(\log \log n)$-bits long. 
This allows us to use a recent deterministic algorithm of Ghaffari and Kuhn \cite{GK20} that produces a $(\deg+1)$-list coloring in $O(\log^2 \card{\colSpace} \cdot \log N)$ 
rounds of {\CONGEST} for a color space of size $\card{\colSpace}$ (using $O(\log \card{\colSpace})$-bit messages).
This improves the randomized complexity of $\Delta+1$-coloring in \CONGEST from $O(\log^5\log n)$ \cite{HKMT21} to $O(\log^3 \log n)$.

We use the following shattering lemma from \cite{CLP20}.

\begin{proposition}[Lemma 4.1 of \cite{CLP20}]
Consider a randomized procedure that generates a subset $\Bad \subseteq  V$ of vertices. Suppose that for each $v \in  V$, we have $\Pr[v \in  \Bad] \leq \Delta^{-3c}$, and this holds even if the random bits outside of the $c$-hop neighborhood of $v$  
are determined adversarially.
W.p.\ $1-n^{-\Omega(c')}$, each connected component in $G[\Bad]$ has size at most $(c' /c)\Delta^{2c} \log_\Delta  n$.
\label{prop:shattering}
\end{proposition}

Since we only apply post-shattering for $\Delta = O(\log^4 n)$ and $c=1$, the connected components of $G[\Bad]$ are of size $n' = \poly(\log n)$. 

We detail now the premises that need to be maintained by our shattering algorithm and how nodes detect in {\CONGEST} if those premises fail.
The properties of ACD (Def.~\ref{def:acd}) are easily verified. 
One can verify that in the ACD-construction of \cite{HKMT21}, the failure probability is $\exp(\Omega(-\sqrt{\Delta}))$. 
We describe shortly how $\zeta_C$ can be approximated in \CONGEST, which is essential for most of the further tests.

Each node in $C$ is supposed to receive slack at least $c \cdot \zeta_C$ from {\slackgeneration}, for some absolute constant $c > 0$.  
If this fails and if $\zeta_C = \Omega(\Delta^{1/3})$, the node is moved to $\Bad$. 
This occurs w.p. $\exp(-\zeta_C) = \exp(-\Omega(\Delta^{1/3}))$.
If $\zeta_C = O(\Delta^{1/3})$, we need not enforce this, since the node should get slack from the put-aside set.

The outliers satisfy the claims of Lemma \ref{lem:removeo} w.p.\ $1-\exp(-\Omega(\Delta))$,
and the claims about the sparsity and degree of non-outliers is easily verified.
The put-aside set $P_C$ is of the given minimum size w.p.\ $1-\exp(-\Omega(\Delta^{1/3}))$ (Lemma \ref{lem:lds}). 
Each node in $C$ has deterministically $|P_C| - a_v = |P_C| - O(\zeta_C)$ neighbors in $P_C$, and gets that much slack (from $P_C$).
After {\synchronizedcolortrial}, we need only ensure that the size of $C$ is at most proportional to the slack of its nodes, or $O(\zeta_C + \Delta^{1/3})$. By Lemma \ref{lem:degred}, this holds with probability $1-\exp(-\Omega(\Delta^{1/3}))$. If it fails, all the nodes in $C$ are added to $\Bad$.
This increases the failure probability only to $\Delta \cdot \exp(-\Omega(\Delta^{1/3})) = \exp(-\Omega(\Delta^{1/3}))$. 

Failure in {\slackcolor} execution occurs when a node terminates without receiving a color. By \cref{lem:slackcolor}, this happens w.p.\ $\exp(-\Omega(\smin^{1/(1+\delta)}))+n^{-\Omega(1)}+\Delta\exp(-\Omega(\smin))$, where $\smin$ is a globally known lower bound on the minimum slack of participating nodes.
When called on $G[\Vsp\cup O]$, the minimum slack is $\Omega(\Delta)$, due to $\core_C$, while when called on $G\setminus P$, it is $\Omega(\Delta^{1/3})$. Since $\delta\leq 1$, the failure probability is less than $\exp(-\Omega(\Delta^{1/6}))$ in both cases.

\paragraph{CONGEST Implementation Issues}
All steps of \cref{alg:logstar} can be implemented in {\CONGEST}, and thus \cref{thm:log-star,thm:shattering} also hold in this model. 
The main hurdle, {\slackcolor} (and its main subroutine, {\multitrial}), is addressed in Sec.~\ref{sec:congest}.
We discuss here the remaining steps.

The leader $w_C$ can be found with a simple $O(1)$-round aggregation procedure within $C$. 
First, choose the node $\ell$ in $C$ with the smallest ID as an interim leader. Then, via a BFS-tree from $\ell$, we can compute
aggregation functions such as $|C|$, $\min_{v\in C} a_v=\min_{v\in C} (|C|-|N_C(v)|)$, etc.
The leader is $w_C = \arg\min_{v\in C} a_v$. 

Also, a constant-factor approximation of $\zeta_C$
can be easily computed and disseminated within $C$ in $O(1)$ rounds of \CONGEST. 
Each node $u\in C$ counts its neighbors in $N_C(w_C)$, and then $\ell$ aggregates $\hat m =m(N_C(w_C))=\frac{1}{2}\sum_{u\in N_C(w_C)}|N(u)\cap N_C(w_C)|$,
 and computes an estimate $\zeta'_C=\frac{1}{\Delta}(\binom{\Delta}{2}-\hat m)$. Note that $m(N(w_C))-\hat m\le e_{w_C}\cdot \Delta$, so $\zeta_C\le \zeta'_C\le \zeta_C+e_{w_C}= O(\zeta_C)$
 (\cref{lem:acdproperties}).

The coloring of $P_C$ is the only step of \cref{alg:logstar} that remains to be explained.
The leader restricts the size of $P_C$ to $\sqrt{|M_C|}/3=\Theta(\sqrt{\Delta})$ which is all we need for \cref{thm:log-star}.
Recall that $\core_C \subseteq N_C(w_C)$.
The leader enumerates the nodes in $\core_C$
and allocates each node $v\in P_C$ a contiguous interval of $2|P_C|+1$ indices, corresponding to a set $R_v$ of nodes.
Since $|\core_C| \ge 2|P_C|^2+|P_C|$, the nodes receive disjoint intervals.
Each node $v\in P_C$ has $a_v = O(\zeta_C)=O(\Delta^{1/3})\le |P_C|$ non-neighbors in $C$ (by \cref{lem:acdproperties}, and assuming $\Delta$ is not too small), and hence it has at least $|R_v| - a_v \ge |P_C|$ neighbors in $R_v$.
Now $v$ can send $|N(v)\cap P_C|+1$ colors from its palette to $w_C$ in $O(1)$ rounds, via the relay nodes in $N(v)\cap R_v$. The topology of $P_C$ can similarly be transmitted. The leader can then properly color $P_C$ locally and forward the colors to the nodes.

\section{List Coloring}
\label{sec:listcol}

Our algorithm works also for \emph{$\Delta+1$-list coloring}, except for one issue: the leader can have a palette that is too different from the rest of the almost-clique. We need only to ensure that enough nodes receive a candidate color in their palette, i.e.,\ derive a counterpart to \cref{obs:leaderpalettenonlist}.
Somewhat surprisingly, just distributing the colors of a leader's palette  suffices, as long as the leader is chosen with this in mind.

In this section, we use the notation $\pal_u$ to denote the initial palette of node $u$, before slack generation. As before, $\pal(u)$ denotes the current palette at any given  time.

The probability that a given node $u \in C$ receives from the leader $w_C$ a candidate color outside its palette is roughly $\frac{|\pal_{w_C} \setminus \pal_{u}|}{|C|}$ (for now, let us ignore the changes due to slack generation).
Define the \emph{discrepancy} of a node $v\in C$ as $\eta_v=\sum_{w\in C}\frac{|\pal_{v}\setminus \pal_{w}|}{|C|}$. We then see that the expected number of nodes that fail to receive a usable color from the leader is $\eta_{w_C}$. 
Discrepancy is also related to the slack of a node: for a given node $v$ and a neighbor $w\in N(v)$, the probability that during slack generation, $w$ picks a color outside $\pal_{v}$ and thus creates a unit slack for $v$ is about $\frac{|\pal_w\setminus \pal_v|}{|C|}$, and the expected slack of $v$ is thus at least $\sum_{w\in C}\frac{|\pal_w\setminus \pal_v|}{|C|}=\eta_v$, where we used that $|\pal_v\setminus \pal_w|=|\pal_w\setminus \pal_v|$. This is called \emph{chromatic slack} and denoted by $\cs_v$. The final crucial observation we make is that \emph{the average discrepancy in $C$ is at most twice the minimum.}

Intuitively, these pieces can be put together as follows. We can pick the minimum discrepancy node as the leader, and move all nodes that deviate much from it to the outlier set $O$; by the min-to-average relation, we only need to move a fraction of the nodes. Now, after internal degree reduction, there will be roughly $\eta_{w_C}+\zeta_C+\Delta^{1/3}$ decolored nodes in $C$ (cf. \Cref{lem:lds}), where $\eta_{w_C}$ accounts for the nodes that did not get a color belonging to its palette. On the other hand, all remaining nodes have proportional slack, so the argument of \Cref{thm:log-star} still holds.
In the remainder of this section, we make this intuition formal.

We modify the algorithm as follows. We compute the smallest anti-degree node $w$ in $C$ as before, and compute $\zeta_C=\zeta_w$ and the sets $O_C$ as before. However, the leader $w_C$ is now the node in $\core_C$ with the smallest value $\cs_{w_C}$.
Also, for convenience we let $O_C\mapsto O_C\cup A_{w_C}\cup X$ (i.e., add $A_{w_C}$ and $X$ to $O_C$), where $X$ is the set of $\eps\Delta$ nodes $v\in C$ with largest $\cs_v$. By the definition of almost-cliques, this reduces $|\core_C|$ by at most $a_{w_C}+\eps\Delta\le 2\eps\Delta$, so $|\core_C|\ge (1/2-4\eps)\Delta$, by Lemma~\ref{lem:removeo}. 
The algorithm is otherwise unchanged.

Let $\eta_C=\min_{v\in C}\eta_v$ be the minimum discrepancy in almost-clique $C$.
\begin{lemma}\label{L:minavg-cs}
The average discrepancy is at most twice the minimum:
$\frac{\sum_{v \in C} \eta_v}{|C|}  \le 2 \eta_C$.
\end{lemma}
\begin{proof}
For all nodes $u,v,w\in C$, we have  $|\pal_u\setminus \pal_v|=|\pal_v\setminus \pal_u|$, since $|\pal_u|=|\pal_v|=\Delta+1$, and also $|\pal_u\setminus \pal_v|\le |\pal_u\setminus \pal_w|+|\pal_w\setminus \pal_v|$. Let $x\in C$ be such that $\eta_x=\eta_C$. Then,
\[
\frac{1}{\card{C}}\sum_{v\in C}\eta_v=\sum_{u,v\in C}\frac{|\pal_u\setminus \pal_v|}{\card{C}^2}\le \sum_{u,v\in C} \frac{|\pal_u\setminus \pal_x| + |\pal_x\setminus \pal_v|}{{\card{C}^2}}=2\sum_{u\in C}\frac{|\pal_x\setminus \pal_u|}{\card{C}}=2\eta_C\ ,
\]
where in the second-to-last equality we used the facts that $|\pal_u\setminus \pal_x|=|\pal_x\setminus \pal_u|$, and that each such term for a fixed $u$ appears in $2|C|$ terms of the left-hand summation.
\end{proof}

The chromatic slack of a node is closely tied to its discrepancy. The proof follows by a standard argument using Talagrand's inequality (cf. \cref{L:sparseGetsSlack}) and is deferred to the appendix.

\begin{restatable}{lemma}{chromaticslack}
\label{lem:chromaticslack}\label{L:cs-diff}
After \slackgeneration, the chromatic slack $\cs_v$ generated for each node $v$ in almost-clique $C$ satisfies
$\pgen\cdot \eta_v/9-a_v \le \cs_v\le 2\pgen\cdot\eta_v$ w.p.\ $1-e^{-\Omega(\eta_v)}$. Moreover, if $\eta_v=O(\log n)$, then $\cs_v=O(\log n)$, w.h.p.
\end{restatable}

We state next an analog of \cref{lem:lds} for list coloring.

\begin{lemma} Let $C$ be an almost-clique, and let $t=\Omega(\eta_{w_C}+\zeta_C)$. W.p.\ $1-\exp(-t)$, the number of decolored nodes of $C$ in step~\ref{st:synchtrial} in \cref{alg:logstar}, is $O(t)$. 
\end{lemma}

\begin{proof}
The proof is nearly identical to that of \cref{lem:lds}, except that for each node $u$, $|\pal(w_C)\setminus \pal(u)|$ is here bounded by $|\pal_{w_C}\setminus \pal_u|+O(\zeta_C)$, using Observation~\ref{obs:leaderpalettenonlist}. Similarly, the expected number of decolored nodes is at most $O(\zeta_C)+\sum_{u\in C}|\pal_{w_C}\setminus \pal_u|=O(\zeta_C+\eta_{w_C})$.
\end{proof}

The proof of our main results (\cref{thm:intrologstar,thm:introshattering}) for list coloring follow by combining the last two lemmas with the arguments of the non-list variants.

\begin{theorem} \label{thm:list-coloring}
Theorems \ref{thm:log-star} and \ref{thm:shattering} hold also for $(\Delta+1)$-list coloring \end{theorem}

The complexity in \CONGEST depends on the size of the color space $\cal C$. The $O(\log^3 \log n)$-round complexity holds if $\card{\colSpace} = \poly(\Delta)$. If $\card{\colSpace}=\poly(n)$ (even when $\Delta = \poly(\log n)$), the right approach is to use network decompositions to redefine the palettes, as done in \cite{HKMT21}, currently resulting in $O(\log^5 \log n)$ complexity.

\section{CONGEST Implementation}
\label{sec:congest}

\subsection{MultiTrial}

The main hurdle in adapting our algorithm to the {\CONGEST} setting concerns the algorithm {\slackcolor} (\cref{alg:slackcoloring})  from \cref{lem:slackcolor}. The centerpiece of the algorithm is the subroutine {\multitrial}, which allows a node  to simultaneously try multiple random colors from its palette.
Let $2 \knuthupuparrow (i+1) = 2^{2 \knuthupuparrow i}$ and $2 \knuthupuparrow 0 = 1$ denote tetration. We show that, under the hypotheses of \cref{lem:slackcolor}, each node can try a number of colors increasing as fast as tetration w.p.\ $1-\exp(-\Omega(\smin^{1/(1+\delta)}))-n^{-\Theta(1)} -\Delta\exp(-\Omega(\smin))$, and therefore gets colored with a similar probability in $O(\log^* \Delta)$ rounds. This success probability simplifies to $1-n^{\Theta(1)}$ when we use this lemma to color large degree graphs (\cref{thm:log-star}), and to $1-\exp(-\Delta^{\Theta(1)})$ when it is used to shatter low degree graphs (\cref{thm:shattering}).

\begin{algorithm}[H]\caption{\slackcolor[($\smin$)], for node $v$} 
\label{alg:slackcoloring}
  \begin{algorithmic}[1]
  \STATE \algorithmicfor\ $O(1)$ rounds \algorithmicdo\  {\tryrandomcolor}($v$).\label{step:slackcolor-begin-init} 
    \STATE \algorithmicif\ $s_v < 2d_v$ \algorithmicthen\ terminate.\label{step:slackcolor-end-init}
    \STATE Let $\sminpow\gets \smin^{1/(1+\delta)}$
    \FOR{$i$ from $0$ to $ \log^* \sminpow$}\label{step:slackcolor-begin-tower}
    \STATE $x_i \gets 2 \knuthupuparrow i$ 
    \STATE $\multitrial(x_i)$ 12 times.
    \STATE \algorithmicif\ $d_v > s_v / \min(2^{x_i},\sminpow^{\delta})$ \algorithmicthen\ terminate.\label{step:slackcolor-termtower}
    \ENDFOR\label{step:slackcolor-end-tower}
    \FOR{$i$ from $1$ to $\ceil*{1/\delta}$}\label{step:slackcolor-begin-finish}
    \STATE $x_i \gets \sminpow^{i \cdot \delta}$ 
    \STATE $\multitrial(x_i)$ 16 times.
    \STATE \algorithmicif\ $d_v > s_v / \min(\sminpow^{(i+1)\cdot\delta},\sminpow)$ \algorithmicthen\ terminate.\label{step:slackcolor-termfinishloop}
    \ENDFOR
    \STATE $\multitrial(\sminpow)$.\label{step:slackcolor-end-finish}
\end{algorithmic}
\end{algorithm}

As its name suggests, \multitrial improves on the success probability of trying a single color by trying up to $\Theta(\log n)$ colors in a single round. While doing so is straightforward in \LOCAL, a na\"ive implementation in \CONGEST would take $\Omega(\log \card{\colSpace})$ rounds for a color space $\colSpace$. We prove that an $O(1)$ round \multitrial procedure can be implemented in \CONGEST. This is achieved by replacing the random sampling of colors by a pseudorandom one. Previously, this was only known to be possible in the very restricted setting of locally sparse graphs~\cite{representativesets}.

To get an intuitive understanding of our approach, let us assume that each node $v$ can sample and communicate to its neighbors a random hash function $h_v : \colSpace \rightarrow [\out]=\set{1,\ldots,\out}$ for a number $\out$ of its choice. To have all nodes try $x$ colors, on each edge $uv$, node $v$ sends to $u$ the hash values of the color it tries through $h_u$ (and reciprocally). If $v$ tries a color $\col$ that hashes to a value different from all the hash values it received, $v$ can safely color itself with $\col$. To make the procedure more efficient, we have 
$v$ pick random colors among those with a hash value $\leq \samp = O(\log n)$ through $h_v$. With this restriction, the neighbors of $v$ only need to tell $v$ about the colors they try that hash to a value $\leq \samp$ through $h_v$. This uses $\samp = O(\log n)$ bits of communication.

For this to work, the hash function must satisfy three properties: first, enough colors must hash to a value $\leq \samp = O(\log n)$; second, collisions must be rare enough for an unique hash to be sampled; and third, it should be possible to communicate a hash function in $O(\log n)$ bits so the process takes $O(1)$ rounds. Increasing $\out$ reduces the number of collisions, but reduces how many elements hash to a value $\leq \samp = O(\log n)$, so a balance must be found. This balance is found at $\out \in \Theta(\card{\pal_v})$.

Assuming the existence of a small enough family of hash functions with the right statistical properties (\cref{lem:representative_hash_functions}), we show how to implement \multitrial efficiently in \CONGEST (\cref{alg:multitrial} and \cref{lem:multitrial-success}). We then prove the existence of the family of hash functions, which we call \emph{representative hash functions}.

For a set $\colSpace$ and a number $\out\in \bbN$, let $[\out]^{\colSpace}$ denote the set of all functions from $\colSpace$ to $[\out]=\{1,\dots,\out\}$. For a function $h$, sets $A,B$, and number $\samp$, let $\hit[\leq \samp]{A}{h}{B}=\set*{\col \in A : h(\col)\in [\samp]\setminus h(B\setminus \set{\col})}$. I.e., $\hit[\leq \samp]{A}{h}{B}$ is the set of elements of $A$ such that: they hash to a value $\leq \samp$ through $h$; and no distinct element in $B$ hashes to the same value. When $\samp$ is clear from the context,  
we simply write $\hit{A}{h}{B}$.

\begin{lemma}
\label{lem:representative_hash_functions}

Let $\alpha,\beta,\nu\in (0,1)$ and $\out\in \bbN$ be s.t.\ $\out \alpha \beta^2 \ln(1/\nu) \geq 2^{17}$, and let $\colSpace$ be a finite set. There exists a family of $\famsize=\Theta\left(\beta \out\nu^{-1} \log\card{\colSpace}\right)$ hash functions $\{h_i\}_{i\in[\famsize]}\subseteq [\out]^{\colSpace}$  and $\samp\leq \out$, $\samp \in \Theta\left(\beta^{-2}\alpha^{-1}\log(1/\nu)\right)$, such that for every $T,P \subseteq \colSpace$ with $\card{T}, \card{P} \in [\alpha \out,\beta \out]$, at least $(1-\nu)\famsize$ of the hash functions $h$ satisfy

\[\card*{\hit[\leq \samp]{T}{h}{P}} \in \frac {\samp \card{T}}{\out} \cdot \range*{1-2\beta,1+\beta}  \]
\end{lemma}

The pseudocode of {\multitrial} is presented in \cref{alg:multitrial}.
Let $\alpha=1/12$, $\beta=1/3$, and for each $\out\in \bbN$, let $\nu_\out = \max(n^{-c},\exp(-2^{-17}\alpha\beta^2\out))$ and $\samp_\out \in \Theta\left(\beta^{-2}\alpha^{-1}\log(1/\nu_\out)\right)$, for a constant $c>3$ (hence, even for $n^2$ events of probability $\nu_\out$, when $\out \in \omega(\log n)$, none occurs w.h.p.). We assume that all the nodes know, for each $\out \in [2\beta^{-1}\Delta]=[6\Delta]$, a common family of hash functions $\HFset^\out=\parens{h^{(\out)}_i}_{i\in[\famsize]}\subseteq [\out]^{\colSpace}$ and value $\samp_\out$ with the properties of Lemma~\ref{lem:representative_hash_functions}. This could be achieved, e.g., by having each node compute the lexicographically first such pair of family and parameter, for each $\out$. Note that $\samp_\out \in O(\log n)$ for all $\out$, and that this parameter can be chosen to be the same $\samp=\Theta(\log n)$ for all values of $\out \in \omega(\log n)$.

\begin{algorithm}[H]\caption{{\multitrial}($x$), for node $v$}
\label{alg:multitrial}
  \begin{algorithmic}[1]
    \STATE Let $\out_v\gets 6\card*{\pal_v}$, pick a random $h_v=h^{(\out_v)}_{i_v}\in \HFset^{\out_v}$, broadcast $\out_v,i_v$ to $N(v)$. 
    \STATE $X_v\gets$ $x$ independently chosen random colors in $\hit{\pal_v}{h_v}{\pal_v}$.\label{st:mulxv}
    \FORALL {$u\in N(v)$ and all $i\in [\samp_{\out_u}]$}
    \STATE \algorithmicif\  $\exists \col \in X_v$, $h_u(\col)=i$ \algorithmicthen\  $b_{v \rightarrow u}[i]\gets 1$
    \STATE \algorithmicelse\  $b_{v \rightarrow u}[i]\gets 0$ \ENDFOR
    \STATE Send $b_{v \rightarrow u}$ and receive $b_{u \rightarrow v}$  to/from $u$, for all $u\in N(v)$.
    \IF {$\exists \col \in X_v$ s.t. $\forall u\in N(v)$, $b_{u \rightarrow v}[h_v(\col)]=0$} 
    \STATE Adopt some such $\col$ as permanent color and broadcast to $N(v)$.
    \ENDIF
\end{algorithmic}
\end{algorithm}

\begin{lemma}
\label{lem:multitrial-success}
For every node $v$, if $x \leq \card{\pal_v}/2\card{N(v)}$, then  
an execution of {\multitrial}$(x)$ colors $v$ with probability $1-(7/8)^{x}-2\nu$, where $\nu \leq e^{-\Theta( \card{\pal_v})} + n^{-\Theta(1)}$, even when conditioned on any particular combination of random choices of the other nodes.
\end{lemma}

\begin{proof}
    Consider $Y_v=\bigcup_{u\in N(v)} X_u$, the set of colors tried by neighbors of $v$.
    Note that $|Y_v|\le x|N(v)|\le |\pal_v|/2\le \out_v/12$ (recall $\out_v=6 \card*{\pal_v}$), and its composition is independent from $v$'s choice of random colors. Letting  $T_v=\pal_v \setminus Y_v$ and $P_v = \pal_v \cup Y_v$, we have $\card{P_v},\card{T_v},\card{\pal_v} \in [\out_v/12,\out_v/3]$, and so, the triplets $(\out_v,P_v,T_v)$ and $(\out_v,P_v,\pal_v)$ satisfy  Lemma~\ref{lem:representative_hash_functions} with our parameters $\alpha,\beta,\nu$. Let $\samp = \samp_{\out_v}$. 
    The lemma implies that w.p.\ $1-\nu$,  $\card*{\hit{\pal_v}{h_v}{\pal_v}} \le (1+\beta)\cdot \samp\card{\pal_v}/\out_v \le 2\samp/9$,  
    and similarly, w.p.\ $1-\nu$, $\card*{\hit{T_v}{h_v}{P_v}} \ge (1-2\beta)\cdot \samp\card{T_v}/\out_v \ge \samp/36$. 
    Since additionally $(\hit{T_v}{h_v}{P_v}) \subseteq (\hit{\pal_v}{h_v}{P_v}) \subseteq (\hit{\pal_v}{h_v}{\pal_v})$, we conclude that $\hit{T_v}{h_v}{P_v}$ forms a $(\samp/36) / (2\samp/9) = 1/8$ fraction of $\hit{\pal_v}{h_v}{\pal_v}$, and  any color randomly picked in $\hit{\pal_v}{h_v}{\pal_v}$ is  in $\hit{T_v}{h_v}{P_v}$ w.p.\ at least $1/8$.
    Hence, conditioned on the $1-2\nu$ probability event that $|\hit{T_v}{h_v}{P_v}|\ge |\hit{\pal_v}{h_v}{\pal_v}|/8$,  
    the $x$ colors randomly picked by $v$ in  $\hit{\pal_v}{h_v}{\pal_v}$ all miss $\hit{T_v}{h_v}{P_v}$ w.p.\ at most $(7/8)^{x}$. As any color found in $\hit{T_v}{h_v}{P_v}$ will be successful for $v$, $v$ gets colored w.p.\ $1-(7/8)^{x}$, conditioned on an event of probability $1-2\nu$. 
\end{proof}

We now prove the existence of representative hash functions (\cref{lem:representative_hash_functions}). We first prove the following claim. We only consider sets $T,P\subseteq \colSpace$ satisfying $ |T|,|P|\in [\alpha\out,\beta\out]$. A hash function is \emph{$(P,T)$-good} if it satisfies the requirement of the lemma for a given pair $(P,T)$.
We bound the probability that a random function is $(P,T)$-good, for a fixed pair $(P,T)$. 
\begin{claim}
Let $h\in [\out]^{\colSpace}$ be chosen uniformly at random. Then $\Pr\event*{h \text{ is }(P,T)\text{-good}}\ge 1-\nu/2$.
\end{claim}
\begin{proof}
For $\col \in T$, let $X_\col,Y_\col$ be  indicator r.v.'s such that $X_\col=1$  iff $h(\col) \in [\samp]$, and $Y_\col=1$ iff $h(\col) \in [\samp]$ and $\exists \col' \in P\setminus \col$, $h(\col)=h(\col')$. Let $Z_\col = X_\col - Y_\col$; note that $Z_\col$ is also binary and is 1 iff $h(\col)\in [\samp]$ and there is no $\col'\in P\setminus \col$ such that $h(\col')=h(\col)$.  Let $X = \sum_{\col \in T} X_\col$, $Y = \sum_{\col \in T} Y_\col$, and $Z = X - Y$. Note that $Z=\card*{\hit[\leq \samp]{T}{h}{P}}$. We have:

\[\Exp\event*{X_\col} = \samp/\out \qquad\qquad\text{and}\qquad\qquad \Exp\event*{Y_\col} = (\samp/\out)\cdot \parens*{1-\parens*{1-1/\out}^{\card{P}-\tau} } \ ,\]
where $\tau=1$ if $\col \in P$ and $\tau=0$ otherwise. Thus, letting $\mu=\Exp\event*{X}$, we have $\mu=\samp|T|/\out\ge \alpha \samp$.
Using the inequality $1-kx\le (1-x)^k$   
(for $n\in \bbZ_+$, $x \in [0,1]$),  
we have:
\[
 \parens*{1-1/\out}^{\card{P}-\tau}\ge  1-(\card{P}-\tau)/\out \ge 1-\beta\ ,
\]
which implies that $\Exp\event*{Y_\col}\le \beta\Exp\event*{X_\col}$, and hence $\Exp\event*{Y}\le \beta\mu$, and  $(1-\beta)\mu\le \Exp\event*{Z}\le \mu$. 

Note that $h$ being $(P,T)$-good is implied by $|Z-\Exp\event*{Z}|\le \beta \mu$. Thus, we want to ensure: 
\begin{equation}\label{eq:goodhash}
\Pr\event*{ \card*{Z-\Exp\event*{Z}}\le \beta \mu} \geq 1-\nu/2\ .
\end{equation}

As $\abs*{Z - \Exp[Z]} \leq \abs*{X - \Exp[X]} + \abs*{Y - \Exp[Y]}$, we derive (\ref{eq:goodhash}) by arguing about the concentration of the variables $X$ and $Y$ around their means. Note that $X$ is $1$-Lipschitz and $1$-certifiable, while $Y$ is $2$-Lipschitz and $2$-certifiable, so we use Talagrand's inequality from Lemma~\ref{lem:talagrand}. Let us set $t = (\beta/4) \cdot \mu$ when applying the Lemma to both $X$ and $Y$ and assume $\samp$ to be large enough to ensure that $ 60\sqrt{2\cdot \mu} < t$, which $\mu > 2\cdot (240/\beta)^{2}$ guarantees, and a fortiori $\samp>2^{17}/(\alpha\beta^2)$ (recall that $\mu\ge \alpha \samp$).
The probability that $\abs*{X-\mu}\ge 2t=(\beta/2) \mu$ is then bounded by $4\exp\parens*{-2^{-6}t^2/\mu }$, where  $t^2/\mu=\beta^2\mu/16\ge \beta^2\alpha \samp/16$.
Taking $\samp = \Theta(\log(1/\nu)/(\beta^2\alpha))$ is enough to ensure that   $|X-\mu|,|Y-\Exp\event*{Y}|\le (\beta/2)\mu$, as well as $|Z-\Exp\event*{Z}|\le 2\cdot (\beta/2) \mu$, hold w.p.\ $1-\nu/2$, as required by (\ref{eq:goodhash}).
\end{proof}

\begin{proof}[Proof of \cref{lem:representative_hash_functions}]
Let  $h_1,\ldots,h_\famsize\in [\out]^{\colSpace}$ be $\famsize$ functions, chosen independently and uniformly at random. For fixed sets $T,P$, let $X_i=1$ if $h_i$ is not $(P,T)$-good,  otherwise $X_i=0$; by the claim above, $\Pr\event*{X_i}\le \nu/2$. By Chernoff (\cref{lem:basicchernoff}), the probability that more than $\nu \famsize$ of them fail to be $(P,T)$-good is
$ \Pr\event*{\sum_{i\in[\famsize]} X_i \geq \nu \famsize} \leq e^{-\nu \famsize/6 }$. There are at most $|\colSpace|^{\beta \out+1}$ choices for each of the subsets $P$ and $T$, so at most $|\colSpace|^{2\beta \out+2}$ choices for the pair $(P,T)$. By the union bound, the probability that there are $\nu \famsize$ functions that are not $(P,T)$-good for some $P,T$, is at most $|\colSpace|^{4\beta \out}e^{-\nu \famsize/6}<1$, assuming $\famsize>(24\beta \out/\nu)\log |\colSpace|$. Thus, there is a family of $\famsize$ hash functions such that for every pair $P,T$, at least $(1-\nu)\famsize$ of them are $(P,T)$-good. 
\end{proof}

\subsection{Proof of \texorpdfstring{\cref{lem:slackcolor}}{Lemma~\ref{lem:slackcolor}}}

We are now ready to complete the analysis of  \cref{alg:slackcoloring} (\slackcolor), proving \cref{lem:slackcolor} via \cref{lem:slackcolor-tower,lem:slackcolor-init,lem:slackcolor-finish}.

\slackcolorlemma*

\Cref{alg:slackcoloring} is naturally decomposed in three phases: a first phase where a loop of \tryrandomcolor increases the slack to degree ratio from a small constant $\iratio$ to $2$; a second phase in which nodes use \multitrial to try a number of colors increasing as fast a tetration between loops; and a third phase in which loops slowly reduce the degree by the slack to degree ratio obtained in previous phases, until the slack to degree ratio becomes small enough that nodes can try the number of colors needed to get successfully colored with the probability of success claimed in \cref{lem:slackcolor}.
 \cref{lem:slackcolor-init} is responsible for the first loop of {\slackcolor}, showing that the probability of a node terminating after the loop is exponentially small in the slack. Similarly, \cref{lem:slackcolor-tower,lem:slackcolor-finish} show that terminating during each of the two subsequent loops is exponentially small in $\smin$.

In the following analysis, we take as a unit of time an \emph{iteration}, which corresponds to an application of {\tryrandomcolor} or {\multitrial} in {\slackcolor}. 
For each lemma, let $d_v$ be the degree of $v$ before some number of iterations, and $d'_v$ the degree of $v$ after them. The degree should be understood as dynamic here: we only count neighbors that are participating in the algorithm (e.g., when coloring outliers and sparse nodes with \slackcolor, nodes out of $O \cup \Vsp$ do not count towards nodes' degrees), and so the degree of a node decreases both when one of its neighbors terminates or gets colored.

\begin{lemma}
    Let $\iratio > 1$. Suppose all nodes satisfy $s_v \geq  d_v/\iratio$. Then after $t=O(\iratio \log \iratio)$ iterations of all nodes running \tryrandomcolor, a node $v$ satisfies $s_v \geq 2d'_v$ w.p.\ $1-\exp(-\Omega(s_v))$. This holds conditioned on arbitrary random choices of nodes at distance $\geq 2$ from $v$.
\label{lem:slackcolor-init}
\end{lemma}
\begin{proof}
    Due to slack, each color try succeeds w.p.\ at least $p_\iratio=(1/\iratio)/(1+1/\iratio)=1/(1+\iratio)$ regardless of the random choices of other nodes. Notably, each color try in $v$'s neighborhood succeeds with at least this probability, regardless of the random choices at distance $\geq 2$ from $v$. In $t$ iterations of {\tryrandomcolor}, each node stays uncolored w.p.\ at most $(1-p_\iratio)^{t}$, hence in expectation, $(1-p_\iratio)^t d_v$ neighbors of $v$ stay uncolored. Setting $t=\iratio\ln(4\iratio)$
    implies $(1-p_\iratio)^{-t}= (1+1/\iratio)^{t}\ge  4\iratio$, and with $\delta=(1-p_\iratio)^{-t}s_v/(2d_v)-1$, we have $\delta\geq 4\iratio \cdot 1/(2\iratio)-1= 1$.  
    The lemma then follows by \Cref{lem:chernoff}: 
    \[\Pr\event*{d'_v \geq \frac {s_v} 2} = \Pr\event*{d'_v \geq (1+\delta)(1-p_\iratio)^{t} \cdot d_v} \leq \exp\parens*{-\frac \delta 3 \cdot (1-p_\iratio)^{t} \cdot d_v} \leq e^{-s_v/12} 
    \ .\qedhere\]
\end{proof}

\begin{lemma}
    Let $v$ be a node and $x\geq 1$ be an integer. Suppose $d_u \leq s_u / x$ and $s_u \geq \smin$ for all $u\in N(v)\cup\set{v}$. Let $y\geq s_v \cdot 2^{-x}$. Then after $t=12$ iterations of \multitrial[$(x)$], $v$ satisfies $d'_v \leq y$ w.p.\ $1-\exp(-\Omega(y))-O(\nu\cdot \Delta)$, where $\nu \leq e^{-\Omega(\smin)}+n^{-\Theta(1)}$. This holds conditioned on arbitrary random choices of nodes at distance $\geq 2$ from $v$.
    \label{lem:slackcolor-tower}
\end{lemma}
\begin{proof}
    First, since $(7/8)^6 < 1/2$, running \multitrial[$(x)$] $t=6 \times 2$ times makes a node get colored w.p.\ at least $1-2^{-2x}$, conditioned on a high probability event (of probability $\geq 1-24\nu$), by \cref{lem:multitrial-success}. Conditioning on such high probability events for all neighbors of $v$, this implies $\Exp\event{d'_v}\le 2^{-2x}d_v\le (2^{-2x}/x)s_v\le y/2$.
    Applying \Cref{lem:chernoff} with $\delta=(y/\Exp\event{d'_v}) - 1 \geq 1$, we get:
    \[\Pr\event*{ d'_v > y } = \Pr\event*{ d'_v > (1+\delta)\Exp\event{d'_v} }\leq \exp\parens*{-(\delta/3)\Exp\event{d'_v} } = \exp\parens*{-\Omega(y)} \]
    
    Therefore, a node that -- together with its neighborhood -- satisfies $d_v \leq s_v / x$, satisfies $d'_v \leq y$ w.p.\ at least $1-\exp(-\Omega(y))-O(\nu \cdot \Delta)$ after $t=12$ iterations of \multitrial[$(x)$]. Since $\multitrial$ succeeds with the claimed probability regardless of the random choices of a node's neighbors, the lemma holds for arbitrary random choices at distance $\geq 2$ from $v$. 
\end{proof}

\begin{lemma}
    Consider a node $v$ and integers $\smin$ and $x\geq \ln(d_v)$ such that each of $v$'s neighbors $u$ satisfies $s_u \geq x \cdot d_u$ and $s_u \geq \smin$.
    Then for every $y\ge 1$, after $16$ iterations of \multitrial[$(x)$], $d'_v \leq y/x$ w.p.\ $1-e^{-y}-O(\nu \cdot \Delta)$ where $\nu \leq e^{-\Omega(\smin)}+n^{-\Theta(1)}$. This holds conditioned on arbitrary random choices of nodes at distance $\geq 2$ from $v$.
\label{lem:slackcolor-finish}
\end{lemma}
\begin{proof}
     Conditioning on $\leq \Delta$ high probability events (of probability $\geq 1-32\nu$) related to \multitrial's success, after $16$ iterations of \multitrial[$(x)$], each neighbor $u$ of $v$ stays uncolored w.p.\ at most $e^{-2x}$. This holds even conditioned on arbitrary random choices from $u$'s neighbors (and so of nodes at distance at least $2$ from $v$). Thus, for a specific set of $k\le d_v$ neighbors of $v$, with the same conditioning, the probability that they all stay uncolored is bounded by $e^{-2k\cdot x}$ (using the chain rule).
     The probability that $k$ or more neighbors of $v$ stay uncolored is bounded by $\binom{d_v}{k} \cdot e^{-2k\cdot x} \leq \exp(k\cdot (\ln d_v-2x)) \leq e^{-k\cdot x}$. So, $d'_v \leq y/x$ holds
    w.p.\ at least $1-e^{-y} - O(\Delta \nu)$.
\end{proof}

\begin{proof}[Proof of \cref{lem:slackcolor}]
    After the first loop of \cref{alg:slackcoloring}, 
    by \cref{lem:slackcolor-init}, each node satisfies $s_v\geq 2d_v$ w.p.\ $1-\exp(-\Omega(s_v)) \geq 1-\exp(-\Omega(\smin))$. After step~\ref{step:slackcolor-end-init}, all non-terminated nodes $v$
    satisfy $s_v\geq 2d_v$. Let $\sminpow=\smin^{1/(1+\delta)}$, as in the algorithm. Note that for every $v$, $s_v\ge \sminpow^{1+\delta}$.
    
    In what follows, let us condition on the $\leq \Delta \log^*\Delta$ high probability events related to \multitrial's success, which all hold with probability $1-O(\Delta\log^*\nu)$ where $\nu \leq e^{-\Omega(\smin)} + n^{-\Theta(1)}$. We add those terms back at the end of the computation.
    
    Let us consider steps~\ref{step:slackcolor-begin-tower} to~\ref{step:slackcolor-end-tower}. Let $x_i=2\knuthupuparrow i$ and $y_i=s_v/\min(2^{x_i},\sminpow^{\delta})$ (note that $2^{x_i}=x_{i+1}$). At the beginning of the $i$th execution of the loop (starting with $i=0$), all nodes satisfy $d_v \leq s_v / x_i$, and by definition $y_i \geq s_v\cdot 2^{-x_i}$. By \cref{lem:slackcolor-tower}, the following execution of \multitrial[$(x_i)$] ensures that a node $v$ passes the test at the end of the $i$th loop w.p.\ $1-\exp(-\Omega(y_i))$. A node $v$ passes all the end-loop tests w.p.\ $1-\sum_i \exp(-\Omega(y_i)) = 1-\exp(-\Omega(s_v/\sminpow^{\delta})) \geq 1-\exp(-\Omega(\sminpow))$. At the end of this loop, each non-terminated node $v$ satisfies  $d_v\le s_v/\sminpow^\delta$.
    
    Finally, in steps~\ref{step:slackcolor-begin-finish} to~\ref{step:slackcolor-end-finish}, each loop execution decreases the degree by a multiplicative factor of $\sminpow^{-\delta}$. More precisely, let $y_i=s_v \cdot \sminpow^{-i\cdot \delta}$. By \cref{lem:slackcolor-finish}, the $i$th execution (starting from $1$) starts with nodes $v$ all satisfying $d_v \le y_i$, and ends each of them satisfying $d_v \le y_{i+1}$ (i.e., passing the test at line \ref{step:slackcolor-termfinishloop}) w.p.\ $1-e^{-\sminpow}$. Nodes that pass all the tests (w.p.\ $\ge 1-(1/\delta)e^{-\sminpow}=1-e^{-\Omega(\sminpow)}$, since $\delta>1/\smin$) end up with $s_v/d_v \geq \sminpow$. Running \multitrial[$(\sminpow)$] at this point, each remaining node gets colored w.p.\  $1-e^{-\Omega(\sminpow)}$. In total, the probability of not getting colored (in this last step or due to an early termination) is  $e^{-\Omega(\sminpow)}$. This holds even conditioned on arbitrary random choices at distance $\geq 2$ from $v$, as all the lemmas we invoked do.
\end{proof}

\subsection{Large colors}
\label{sec:large-colors}

We have implicitly assumed until now that sending a color over an edge, as nodes do when broadcasting their permanent color to their neighbors, only takes $O(1)$ rounds. This is possible if the color space $\colSpace$ is of size $\card{\colSpace} \in n^{O(1)}$. In \cref{lem:representative_hash_functions}, the dependency of $t$ in $\card{\colSpace}$ is only $\log \card{\colSpace}$, meaning that sending a representative hash function still only takes $O(1)$ rounds even for $\card{\colSpace} \in \exp(n^{\Theta(1)})$. Can we tolerate such a large color space in other parts of the algorithm? We resolve this in the affirmative.

We achieve this using a family $\HFset$ of $1+\eps$-approximately universal hash functions, i.e., a set of hash functions $h:[N] \rightarrow [M]$ such that for all $x_1 \neq x_2$, $\Pr_{h\gets\HFset}\event{h(x_1)=h(x_2)}\leq (1+\eps) / M$. There exists small enough families of such hash functions so that specifying an element in the family only takes $O(\log\log N + \log M + \log(1/\eps))$ bits (\cite{BJKS93}, or Problem 3.4 in~\cite{Vadhan12}). Set $\eps=1$ and let us hash to $M=\Theta(n^{d})$ values, where $d \in \Theta(1)$. Under these assumptions, sending an hash value only takes $O(1)$ rounds, and sending an element of $\HFset$ takes $O(\ceil{\log\log\colSpace / \log n})$ rounds -- in particular, $O(1)$ if colors are written on $\poly(n)$ bits. Let each node $v$ pick and broadcast a random $1+\eps$-approximately universal hash function $h_v$ from $\HFset$ at the start of our algorithms. Whenever a node $u$ was previously sending a color $\col$ to a node $v$ in our algorithms, we now have $u$ send $h_v(\col)$ to $v$. Granted no collision occurs in any neighborhood, these hash values perfectly replace the actual colors wherever nodes were previously using the exact colors of their neighbors, such as when updating their palettes, computing their chromatic slack, and gathering at $w_C$ palette colors from each node in $P_C$ (using $h_{w_C}$ to hash all colors in that last case).

We now ensure no collision occurs in any neighborhood w.h.p., by taking $d$ appropriately large. Consider a node $v$ and its neighborhood. There are at most $(\Delta+1)^2$ distinct colors in the palettes of $v \cup N(v)$. The probability that a collision occurs in these colors with a random hash function from $\HFset$ is bounded by $\binom{(\Delta+1)^2}{2} \cdot n^{-d} \leq n^{-d+4}$. So, w.p.\ at least $1-n^{-d+5}$, there are no collisions in all neighborhoods. Setting $d \ge 6$, this holds w.h.p.

\section{Refinements}
\label{sec:refinements}

We present in this section various refinements in the round complexity, number of colors, shattering parameters, and limitations.

\subsection{Color and Time Tradeoffs}
We can get log-star running time if we have slightly more colors.
Namely, we argued in Lemma \ref{lem:degred} that after \synchronizedcolortrial, each node has $O(\zeta_C + \log n)$ uncolored neighbors, and it has slack $\Theta(\zeta_C)$ when $\zeta_C = \Omega(\log n)$. Thus, if given additional slack $\Theta(\log^{1+\delta} n)$, all nodes satisfy the conditions for \cref{lem:slackcolor} and can be colored in $O(\log^* n)$ rounds. 
\begin{corollary}
\label{cor:cor_of_main1}
There is a randomized distributed $(\Delta+O(\log^{1+\delta} n))$-list coloring algorithm with $O(\log^* n)$ runtime, for any $\delta > 0$. 
\end{corollary}

Chang, Li, and Pettie \cite{CLP20} showed that log-star running time was possible for $\Delta + O(\log^\gamma n)$-coloring all graphs, or for $\Delta+1$-coloring graphs of sparsity $\Omega(\log^\gamma n)$, for some large constant $\gamma$. They state that the constant $\gamma$ was impractically large and that it would be useful to know if the results holds when $\gamma$ is small, say 1. Our result shows that they hold for any $\gamma > 1$. In fact, it can be extended to $\gamma = 1 + 1/\Omega(\log^* n)$.

\smallskip
We next show that we need fewer colors on graphs with a clique number $o(\Delta)$, since they have non-trivial sparsity.

Recall that each node $v$ of sparsity $\zeta_v = \Omega(\log n)$ obtains slack $c \cdot \zeta_v$ after slack generation, w.h.p. 
It suffices to use, say, half of that slack to make color tries easier, which means that $v$ could do with a smaller palette of size $\Delta - c\zeta_v/2$.
We can perform slack generation using the initial palette $\pal_0 = \{1,2,\ldots, \Delta/2\}$, and decrease the parameter $\pgen$ of {\slackgeneration} by a factor of 2. Each node $v$ can determine the slack $s_v$ that it obtains and then use in the rest of the algorithm the palette $\pal'_v = \{1,2,\ldots, \Delta - s_v/2\}$. The algorithm is otherwise unchanged and only the constant-factors of the analysis are affected. 

We can observe the following lower bound on the minimum node sparsity $\zeta$ in terms of the clique number $\omega$.

\begin{lemma}
For each node $v\in V$, $\zeta_v\ge \Delta/\omega-1$, where $\omega$ is the clique number of $G$. 
\end{lemma}
\begin{proof} 
Let $v$ be a node, $H = G[N(v)]$ be the graph induced by its neighborhood, $d=|N(v)|$, and let $\overline m$ be the number of non-adjacent pairs of nodes in $N(v)$.
By Turan's theorem, $\omega>\alpha(\overline H) \ge d^2 / (2\overline{m}+d)$.
Thus, $\overline{m} \ge d^2/\omega-d$, and $m(N(v))\le \binom{d}{2}-\overline m \le \Delta^2/2-\Delta^2/\omega+\Delta/2$. The sparsity of $v$ is then $\zeta_v=(\Delta-1)/2-m(N(v))/\Delta\ge\Delta/\omega-1$.
\end{proof}

\begin{theorem}\label{thm:cliquecol}
There is a randomized distributed coloring algorithm with $O(\log^* n)$ runtime that uses $\Delta - \Omega(f)$ colors, for any $f \le \Delta/\omega$ with $f = \log^{1+\Omega(1)} n$. 
\end{theorem}

The only related result we know of is a $\Delta - \Omega(\Delta/\chi)$-coloring algorithm running in time $O(\log \chi + \log^* n)$ for $\Delta/\chi = \Omega(\log^{1+\Omega(1/\log^* n)} n)$ by Schneider, Elkin, and Wattenhofer \cite{SEW13}, where $\chi$ is the chromatic number of the graph.

\subsection{Improvements and Limits for High-Degree Coloring} 
\label{sec:betterhideg}

We can relax the precondition of \cref{thm:log-star} on the degree to $\Delta \ge \log^{2+\delta} n$, for any constant $\delta > 0$. This is obtained by performing the sample-and-delete task of \disjointsample more gradually,
thereby maintaining better tradeoffs between the sample size and the dependency degree needed to apply the read-$k$ concentration bound (\cref{lem:kread}).

We observe that in essence, our task  is finding an \emph{independent transversal} in a graph derived from $G$. We start by stating our result in terms of transversals, since this may be of independent interest, then explain how it applies to our coloring algorithm.
The $k$-independent transversal problem takes as input a graph $H=(V_H,E_H)$, partitioned into independent sets $I_1, I_2, \ldots, I_t$, and the objective is to find an independent set $P \subset V_H$ such that $\card{I_i \cap P} \ge k$, for all $i$.
The primary parameters besides $k$ are the maximum degree $\Delta=\Delta_H$ and the size of the smallest set $I_i$. A celebrated result of Haxell \cite{Haxell01} shows that every graph has a 1-independent transversal when $|I_i| \ge 2 \Delta$, for all $i$, and this is best possible.
We show below how to find a $k$-independent transversal in $O(1)$ rounds of \CONGEST under the assumption that $|I_i|\ge ck\Delta$, for a large enough constant $c$, and $k=\Delta^{\delta/(1+\delta)}\log n$,  for any constant $\delta\in (0,1)$, where $n=|V_H|$. We assume, for simplicity, that $\delta=1/m$ is the inverse of an integer, although the proof is easy to adapt to any rational value. 

Observe the main difference of this algorithm from \disjointsample: rather than keeping only nodes with no sampled neighbors, we keep the ones with few sampled neighbors, and refine them further.

\begin{algorithm}[H]\caption{{\lowdegreesample}($P$, $q$, $B$)}
\label{alg:lowtrans}
  \begin{algorithmic}[1]
  \STATE $S\gets$ each node $v\in P$ is sampled independently w.p.\ $\pdisj = 1/(2q)$
  \RETURN $P' \gets \{v\in S : |N(v)\cap S| < B/q\}$
 \end{algorithmic}
\end{algorithm}

\begin{algorithm}[H]\caption{{\transversal}($\delta$)} 
\label{alg:transversal}
  \begin{algorithmic}[1]
    \STATE  
    $m \gets 1/\delta$, $q\gets\Delta^{\delta/(1+\delta)}$, $B_0 \gets \Delta$, $P_0\gets V_H$
    \FOR{$j=1$ to $m+1$} 
      \STATE $P_j \gets$ {\lowdegreesample}($P_{j-1}$, $q$, $B_{j-1}$)
      \STATE $B_j \gets B_{j-1}/q = \Delta^{1-j\delta/(1+\delta)}$
    \ENDFOR
    \RETURN $P_{m+1}$
 \end{algorithmic}
\end{algorithm}

\begin{lemma}\label{lem:lowdeganalysis} 
Let numbers $B,q>0$ and set $P$  of vertices be such that for every $i\in [t]$, $|P\cap I_i| \ge c qB \log n$,  for a sufficiently large constant $c$, and $|N(v)\cap P|\le B$, for every $v\in I_i$.
Let $P' =$ \lowdegreesample[$(P,q,B)$].
Then, $|P'\cap I_i| \ge |P\cap I_i|/(8q)$, w.h.p.\ for all $i\in [t]$.
\end{lemma}
\begin{proof}
    Let $S$ be the sampled set  in \lowdegreesample[($P,q,B$)], and let $I=I_i$, for some $i\in [t]$.
    Observe that by Lemma~\ref{lem:basicchernoff} with  $q_i=p_{s}=1/(2q)$, we have $|S\cap I| \ge |P\cap I|/(4q)\ge (c/4)B\log n$, w.h.p.  The remainder of the proof is conditioned on this event.
    
    Let $Y_w$, $w\in P$, be the independent indicator random variable of the event that $w \in S$, and let $X_v$, $v\in I$, be the indicator random variable of the event that  $|N(v)\cap S| \ge B/q$.
    Since $|N(v)\cap P|\le B$,
    $\Exp[|N(v)\cap S|]\le B \cdot p_{s} = B/(2q)$.
By Markov, $\Pr[X_v=1]\le \Exp[|N(v)\cap S|]/(B/q) \le 1/2$. 
Note that each variable $X_v$ is a function of independent variables $Y_w$, for $w\in N(v)$, and each $Y_w$ influences at most $|N(w)\cap P|\le B$ of the variables $X_{v}$; 
thus, for a given sample $S$, $\{X_v\}_{v\in S\cap I}$ is a read-$B$ family of random variables, and by Lemma~\ref{lem:kread}, $X_I = \sum_{v\in S\cap I}X_v \le |S\cap I|/2$ holds w.p.\  $1-\exp(-\Omega(|S\cap I|/B)) = 1 - \exp(-\Omega(c\log n))$, recalling $|S\cap I| \ge   (c/4)B\log n$.  We choose the constant $c$ large enough, so that the bound holds w.h.p.; then, at least $|S\cap I| - X_I \ge |S\cap I|/2 \ge |P\cap I|/(8q)$ nodes have  degree at most $B/q$ in $S$, as claimed. 
\end{proof}

\begin{theorem}
Let $\delta\in (0,1)$. Consider an instance $H$ with a partition $\{I_i\}_{i\in [t]}$, where for all $i\in [t]$, $|I_i|\ge ck\Delta$, for a large enough constant $c$, and $k\ge \Delta^{\delta/(1+\delta)}\log n$.
Then $\transversal(\delta)$  
returns a $k$-independent transversal, w.h.p.   
\end{theorem}

\begin{proof}
Let $P = P_{m+1}$ be the set output by \transversal, and let $I=I_i$, for some $i$.
The last iteration, $m+1$, has $B_{m+1}/q = \Delta^0 = 1$. Thus, by construction, $P$ is a transversal. To prove the size bound, we apply
Lemma \ref{lem:lowdeganalysis} and the union bound to get $|P_j\cap I| \ge |P_{j-1}\cap I|/(4q)$, for each $j=1, 2, \ldots, m+1$, and thus 
\[
|P\cap I| = |P_{m+1}\cap I| \ge \frac{|P_0\cap I|}{(4q)^{m+1}} = \frac{|I|}{(4q)^{m+1}}\ge \frac{ck\Delta}{4^{1+1/\delta}\Delta}=(c4^{-1-1/\delta})k\ .
\]
To apply Lemma~\ref{lem:lowdeganalysis}, we need $|P_m\cap I|\ge c' qB_m\log n$, for a large enough $c'$. Note that $B_m=q$, while the calculation above shows that $|P_m\cap I|\ge (c4^{-1-1/\delta})kq=(c4^{-1-1/\delta})qB_m\log n$.
\end{proof}

To apply this to our coloring setting, we let $H$ be the subgraph of $G$ induced by $\bigcup_{C: \zeta_C\le B_0} \core_C$, where the union is over all almost-cliques with sparsity $\zeta_C\le B_0=O(\log^{1+\delta} n)$, and we remove all edges within each $C$. Thus, we have the correspondence $I_i\gets \core_{C_i}$, where $C_i$ is the $i$th such almost-clique, 
and the degree of a node in $H$ is (at most) its external degree in $G$. Since we apply the procedure to almost-cliques $C$ with $\zeta_C=O(\log^{1+\delta} n)$, the latter also bounds the external degree of nodes, that is, the degree in $H$. We let $k=\Theta(\log^{1+\delta} n)$, and so we only need $|I_i|=\Omega(k\cdot \log^{1+\delta} n)=\Omega(\log^{2+2\delta} n)$. Thus \transversal allows us to sample put-aside sets $P_C$ of size $\Omega(\log^{1+\delta} n)$ in cliques of sparsity $O(\log^{1+\delta} n)$ when the maximum degree $\Delta$ of $G$ is $\Omega(\log^{2+2\delta} n)$. Replacing {\disjointsample} by this alternative procedure in Alg.~\ref{alg:logstar} is the only modification to the algorithm.

We state as conclusion the following improvement of \cref{thm:log-star}.

\begin{corollary}
\label{cor:log-star}
There is a randomized \CONGEST $\Delta+1$-coloring algorithm with runtime $O(\log^* n)$, for graphs with $\Delta =\Omega(\log^{2+\delta} n)$, for any constant $\delta>0$.
\end{corollary}

\paragraph{Limitation result} 
The question if the degree lower bound of \cref{thm:log-star} can be further decreased is open. We note here that our transversal construction is nearly tight. 
We show via the probabilistic method that there is a graph $H$ with a partition $V_H=I_1\cup\dots,\cup I_t$ such that $|I_i|=\Theta(k\Delta/\log (k\Delta))$, and $H$ has no $k$-independent transversal. 

Let $k\ge 2,t\ge 2,\Delta\ge 64+12\ln(kt)$ be integers. To construct $H$, let $|I_i|=D$, where $D$ is the largest integer such that $D< k\Delta/(16\ln D)$; note that $D=\Theta(k\Delta/\ln(k\Delta))$ and $D\ge 4$. Each edge between different parts $I_i,I_j$ is drawn independently, w.p.\ $p=\Delta/(2n)$, where $n=|V_H|=Dt$. By Chernoff bound and union bound, w.p.\ $1-n e^{-\Delta/6}>1/2$,
 the degree of each node is at most $\Delta$, where we used $\ln n-\Delta/6\le \ln (Dt)-\Delta/6\le \ln(k\Delta t)-\Delta/6\le -1$.
 For every subset $S\subseteq V_H$ such that $|S\cap I_i|=k$, $1\le i\le t$, the probability that $H[S]$ contains no edges is $(1-p)^{\binom{t}{2}k^2}\le e^{-pk^2t(t-1)/2}=e^{-k^2(t-1)\Delta/(4D)}<e^{-4k(t-1)\ln D}\leq D^{-2kt}$. 
 The number of such subsets $S$ is $\binom{D}{k}^t<D^{kt}$, so by the union bound, the probability that there exists a subset $S$ with the desired property is at most $D^{-kt}\le 1/16$. Thus, w.p.\  
 $1-ne^{-\Delta/6}-D^{-kt}>0$,  
 $H$ has maximum degree at most $\Delta$, and contains no $k$-independent transversal; in particular, such $H$ exists.

In the setting of our coloring algorithm, this means that in order to create slack $k=\log^{1+\delta} n$ corresponding to the external degree $\Delta_H=\log^{1+\delta} n$, we must have $\Delta_G\approx D =\Omega(k\Delta_H/\log(k\Delta_H))=\Omega(\log^{2+\delta'} n)$, for some $\delta'\in (0,\delta)$.

\subsection{Edge coloring}

We note that the \multitrial procedure presented in this work (\cref{alg:multitrial}) is easily adapted to the edge-coloring setting when $\Delta \in \Omega(\log n)$, as is done in~\cite{representativesets}. Indeed, all edges have sparsity $\Theta(\Delta)$ in the edge-coloring setting, so when $\Delta \in \Omega(\log n)$ we can generate $\Omega(\Delta)$ slack for all edges w.h.p.\ (chromatic slack, when adjacent palettes significantly differ, usual sparsity-based slack when not). The edges need not even know the size of their current palette when we run \cref{lem:slackcolor}, as this size is guaranteed to be at least a constant fraction of the original palette throughout, and it suffices that only one of each edge's endpoints is given the edge's list of colors at the beginning of the algorithm. This yields the following results.

\begin{theorem}
There are $O(\log^* n)$-round randomized \CONGEST algorithms for $2\Delta-1$-list edge coloring for graphs with $\Delta = \log^{1+\Omega(1)} n$ and $2\Delta+ \log^{1+\Omega(1)} n$-list edge coloring arbitrary graphs.
\end{theorem}

We leave improving the \CONGEST randomized complexity of the list edge-coloring problem with lists of size $O(\Delta)$ when $\Delta \in \log^{1+O(1)} n$ to future works.

\bibliographystyle{plain}
\bibliography{references}

\begin{thebibliography}{10}

\bibitem{alon86}
Noga Alon, László Babai, and Alon Itai.
\newblock A fast and simple randomized parallel algorithm for the maximal
  independent set problem.
\newblock {\em J.\ of Algorithms}, 7(4):567--583, 1986.

\bibitem{ACK19}
Sepehr Assadi, Yu~Chen, and Sanjeev Khanna.
\newblock {Sublinear algorithms for {$(\Delta + 1)$} vertex coloring}.
\newblock In {\em {Proceedings of the ACM-SIAM Symposium on Discrete Algorithms
  (SODA)}}, pages 767--786, 2019.
\newblock Full version at arXiv:1807.08886.

\bibitem{awerbuch89}
Baruch Awerbuch, Andrew~V. Goldberg, Michael Luby, and Serge~A. Plotkin.
\newblock Network decomposition and locality in distributed computation.
\newblock In {\em {Proceedings of the Symposium on Foundations of Computer
  Science (FOCS)}}, pages 364--369, 1989.

\bibitem{BKM19}
Philipp Bamberger, Fabian Kuhn, and Yannic Maus.
\newblock Efficient deterministic distributed coloring with small bandwidth.
\newblock In {\em {Proceedings of the ACM Symposium on Principles of
  Distributed Computing (PODC)}}, 2020.

\bibitem{Barenboim16}
Leonid Barenboim.
\newblock Deterministic ({\(\Delta\)} + 1)-coloring in sublinear (in
  {\(\Delta\)}) time in static, dynamic, and faulty networks.
\newblock {\em {Journal of the ACM}}, 63(5):47:1--47:22, 2016.

\bibitem{BEG18}
Leonid Barenboim, Michael Elkin, and Uri Goldenberg.
\newblock {Locally-Iterative Distributed ({\(\Delta\)}+ 1)-Coloring below
  Szegedy-Vishwanathan Barrier, and Applications to Self-Stabilization and to
  Restricted-Bandwidth Models}.
\newblock In {\em {Proceedings of the ACM Symposium on Principles of
  Distributed Computing (PODC)}}, pages 437--446, 2018.

\bibitem{BarenboimEK14}
Leonid Barenboim, Michael Elkin, and Fabian Kuhn.
\newblock {Distributed (Delta+1)-Coloring in Linear (in Delta) Time}.
\newblock {\em {SIAM Journal on Computing}}, 43(1):72--95, 2014.

\bibitem{BEPSv3}
Leonid Barenboim, Michael Elkin, Seth Pettie, and Johannes Schneider.
\newblock The locality of distributed symmetry breaking.
\newblock {\em {Journal of the ACM}}, 63(3):20:1--20:45, 2016.

\bibitem{BJKS93}
J{\"{u}}rgen Bierbrauer, Thomas Johansson, Gregory Kabatianskii, and Ben J.~M.
  Smeets.
\newblock On families of hash functions via geometric codes and concatenation.
\newblock In {\em {Advances in Cryptology - {CRYPTO}}}, volume 773 of {\em
  Lecture Notes in Computer Science}, pages 331--342, 1993.

\bibitem{CKP19}
Yi{-}Jun Chang, Tsvi Kopelowitz, and Seth Pettie.
\newblock An exponential separation between randomized and deterministic
  complexity in the {LOCAL} model.
\newblock {\em {SIAM Journal on Computing}}, 48(1):122--143, 2019.

\bibitem{CLP18}
Yi{-}Jun Chang, Wenzheng Li, and Seth Pettie.
\newblock An optimal distributed ({\(\Delta\)}+1)-coloring algorithm?
\newblock In {\em {Proceedings of the ACM Symposium on Theory of Computing
  (STOC)}}, pages 445--456, 2018.

\bibitem{CLP20}
Yi-Jun Chang, Wenzheng Li, and Seth Pettie.
\newblock Distributed {($\Delta+1$)}-coloring via ultrafast graph shattering.
\newblock {\em {SIAM Journal on Computing}}, 49(3):497--539, 2020.

\bibitem{Doerr2020}
Benjamin Doerr.
\newblock {\em Probabilistic Tools for the Analysis of Randomized Optimization
  Heuristics}, pages 1--87.
\newblock Springer International Publishing, Cham, 2020.

\bibitem{DP09}
Devdatt~P. Dubhashi and Alessandro Panconesi.
\newblock {\em Concentration of Measure for the Analysis of Randomized
  Algorithms}.
\newblock Cambridge University Press, 2009.

\bibitem{EPS15}
Michael Elkin, Seth Pettie, and Hsin{-}Hao Su.
\newblock (2{\(\Delta-1\)})-edge-coloring is much easier than maximal matching
  in the distributed setting.
\newblock In {\em {Proceedings of the ACM-SIAM Symposium on Discrete Algorithms
  (SODA)}}, pages 355--370, 2015.

\bibitem{FHK}
Pierre Fraigniaud, Marc Heinrich, and Adrian Kosowski.
\newblock Local conflict coloring.
\newblock In {\em {Proceedings of the Symposium on Foundations of Computer
  Science (FOCS)}}, pages 625--634, 2016.

\bibitem{kread}
Dmitry Gavinsky, Shachar Lovett, Michael Saks, and Srikanth Srinivasan.
\newblock A tail bound for read-{$k$} families of functions.
\newblock {\em Random Structures \& Algorithms}, 47(1):99--108, 2015.

\bibitem{Ghaffari2019}
Mohsen Ghaffari.
\newblock Distributed maximal independent set using small messages.
\newblock {\em {Proceedings of the ACM-SIAM Symposium on Discrete Algorithms
  (SODA)}}, pages 805--820, 2019.

\bibitem{GGR20}
Mohsen Ghaffari, Christoph Grunau, and Václav Rozhoň.
\newblock Improved deterministic network decomposition.
\newblock In {\em {Proceedings of the ACM-SIAM Symposium on Discrete Algorithms
  (SODA)}}, 2021.

\bibitem{GK20}
Mohsen Ghaffari and Fabian Kuhn.
\newblock Deterministic distributed vertex coloring: Simpler, faster, and
  without network decomposition.
\newblock {\em {Computing Research Repository}}, arXiv:2011.04511, 2020.

\bibitem{HKMT21}
Magn{\'{u}}s~M. Halld{\'{o}}rsson, Fabian Kuhn, Yannic Maus, and Tigran
  Tonoyan.
\newblock Efficient randomized distributed coloring in {CONGEST}.
\newblock In {\em {Proceedings of the ACM Symposium on Theory of Computing
  (STOC)}}, 2021.
\newblock Full version at arXiv:2012.14169.

\bibitem{representativesets}
Magn\'us~M. Halld\'orsson and Alexandre Nolin.
\newblock Superfast coloring in {CONGEST} via efficient color sampling.
\newblock In {\em {Proceedings of the International Colloquium on Structural
  Information and Communication Complexity (SIROCCO)}}, 2021.
\newblock Full version at arXiv:2102.04546.

\bibitem{HSS16}
David~G Harris, Johannes Schneider, and Hsin-Hao Su.
\newblock Distributed {$\Delta+1$}-coloring in sublogarithmic rounds.
\newblock In {\em {Proceedings of the ACM Symposium on Theory of Computing
  (STOC)}}, pages 465--478, 2016.

\bibitem{HSS18}
David~G. Harris, Johannes Schneider, and Hsin-Hao Su.
\newblock {Distributed {($\Delta + 1$)}-coloring in sublogarithmic rounds}.
\newblock {\em {Journal of the ACM}}, 65(4):19:1--19:21, 2018.

\bibitem{Haxell01}
Penny~E. Haxell.
\newblock A note on vertex list colouring.
\newblock {\em Comb. Probab. Comput.}, 10(4):345--347, 2001.

\bibitem{johansson99}
{\"{O}}jvind Johansson.
\newblock Simple distributed {$\Delta+1$}-coloring of graphs.
\newblock {\em Inf. Process. Lett.}, 70(5):229--232, 1999.

\bibitem{KuhnW06}
Fabian Kuhn and Roger Wattenhofer.
\newblock On the complexity of distributed graph coloring.
\newblock In {\em {Proceedings of the ACM Symposium on Principles of
  Distributed Computing (PODC)}}, pages 7--15, 2006.

\bibitem{linial87}
Nathan Linial.
\newblock Distributive graph algorithms -- global solutions from local data.
\newblock In {\em {Proceedings of the Symposium on Foundations of Computer
  Science (FOCS)}}, pages 331--335, 1987.

\bibitem{linial92}
Nathan Linial.
\newblock Locality in distributed graph algorithms.
\newblock {\em {SIAM Journal on Computing}}, 21(1):193--201, 1992.

\bibitem{luby86}
M.~Luby.
\newblock A simple parallel algorithm for the maximal independent set problem.
\newblock {\em {SIAM Journal on Computing}}, 15:1036--1053, 1986.

\bibitem{MT20}
Yannic Maus and Tigran Tonoyan.
\newblock Local conflict coloring revisited: Linial for lists.
\newblock In {\em {Proceedings of the International Symposium on Distributed
  Computing (DISC)}}, pages 16:1--16:18, 2020.

\bibitem{molloy2013coloring}
M.~Molloy and B.~Reed.
\newblock {\em Graph colouring and the probabilistic method}, volume~23.
\newblock Springer Science \& Business Media, 2013.

\bibitem{panconesi1992improved}
Alessandro Panconesi and Aravind Srinivasan.
\newblock Improved distributed algorithms for coloring and network
  decomposition problems.
\newblock In {\em {Proceedings of the ACM Symposium on Theory of Computing
  (STOC)}}, pages 581--592, 1992.

\bibitem{Reed98}
Bruce~A. Reed.
\newblock {\(\omega\)}, {\(\Delta\)}, and {\(\chi\)}.
\newblock {\em J. Graph Theory}, 27(4):177--212, 1998.

\bibitem{RG19}
V{\'{a}}clav Rozho\v{n} and Mohsen Ghaffari.
\newblock Polylogarithmic-time deterministic network decomposition and
  distributed derandomization.
\newblock In {\em {Proceedings of the ACM Symposium on Theory of Computing
  (STOC)}}, pages 350--363, 2020.

\bibitem{SEW13}
Johannes Schneider, Michael Elkin, and Roger Wattenhofer.
\newblock Symmetry breaking depending on the chromatic number or the
  neighborhood growth.
\newblock {\em Theor. Comput. Sci.}, 509:40--50, 2013.

\bibitem{SW10}
Johannes Schneider and Roger Wattenhofer.
\newblock A new technique for distributed symmetry breaking.
\newblock In {\em {Proceedings of the ACM Symposium on Principles of
  Distributed Computing (PODC)}}, pages 257--266. {ACM}, 2010.

\bibitem{Vadhan12}
Salil~P. Vadhan.
\newblock Pseudorandomness.
\newblock {\em Found. Trends Theor. Comput. Sci.}, 7(1-3):1--336, 2012.

\end{thebibliography}

\appendix

\section{Concentration Bounds}
\label{app:concentration}

\begin{lemma}[Chernoff bounds]\label{lem:basicchernoff}
Let $\{X_i\}_{i=1}^r$ be a family of independent binary random variables with $\Pr[X_i=1]=q_i$, and let $X=\sum_{i=1}^r X_i$. For any $\delta>0$, $\Pr[|X-\Exp[X]|\ge \delta\Exp[X]]\le 2\exp(-\min(\delta,\delta^2) \Exp[X]/3)$.
\end{lemma}

We use the following variants of Chernoff bounds for dependent random variables. The first one is obtained, e.g., as a corollary of Lemma 1.8.7 and Thms.\ 1.10.1 and 1.10.5 in~\cite{Doerr2020}.

\begin{lemma}[Domination \cite{Doerr2020}]\label{lem:chernoff}
Let $\{X_i\}_{i=1}^r$ be binary random variables, and $X=\sum_i X_i$.
    If $\Pr[X_i=1\mid X_1=x_1,\dots,X_{i-1}=x_{i-1}]\le q_i\le 1$, for all $i\in [r]$ and $x_1,\dots,x_{i-1}\in \{0,1\}$ with $\Pr[X_1=x_1,\dots,X_r=x_{i-1}]>0$, then for any $\delta>0$,
    \[\Pr\event{X\ge(1+\delta)\sum_{i=1}^r q_i}\le \exp\parens{-\frac{\min(\delta,\delta^2)}{3}\sum_{i=1}^r q_i}\ .\]
\end{lemma}

A set $\{X_i\}_{i=1}^r$ of binary random variables is \emph{read-$k$} if there is a set $\{Y_j\}_{j=1}^{m}$ of $m$ independent binary random variables and subsets $\{P_i\}_{i=1}^r$ of indices, $P_i\subseteq [m]$, such that $X_i$ is a function of only $\{Y_j\}_{j\in P_i}$, $i\in [r]$, while for each $j\in [m]$, $|\{i : j\in P_i\}|\le k$. In words, each $Y_j$ influences at most $k$ variables $X_i$. 
\begin{lemma}[read-$k$ r.v.'s \cite{kread}]\label{lem:kread}
Let $\{X_i\}_{i=1}^r$ be a family  of read-$k$ binary random variables with $\Pr[X_i=1]=q_i$, and let $X=\sum_{i=1}^r X_i$, $q=\sum_{i=1}^{r} q_i/r$. For any $\delta>0$,
$\Pr[\card{X -qr} \ge \delta r]\le 2\exp(-\min(\delta, \delta^2) r/(3k))$.
\end{lemma}

A function $f(x_1,\ldots,x_n)$ is  \emph{$c$-Lipschitz} iff changing any single $x_i$ affects the value of $f$ by at most $c$, and $f$ is  \emph{$r$-certifiable} iff whenever $f(x_1,\ldots,x_n) \geq s$ for some value $s$, there exist $r\cdot s$ inputs $x_{i_1},\ldots,x_{i_{r\cdot s}}$ such that knowing the values of these inputs certifies $f\geq s$ (i.e., $f\geq s$ whatever the values of $x_i$ for $i\not \in \{i_1,\ldots,i_{r\cdot s}\}$).
\begin{lemma}[Talagrand's inequality~\cite{DP09}]
\label{lem:talagrand}
Let $\{X_i\}_{i=1}^n$ be $n$ independent random variables and $f(X_1,\ldots,X_n)$ be a $c$-Lipschitz $r$-certifiable function; then for $t\geq 1$,
\[\Pr\event*{\abs*{f-\Exp[f]}>t+30c\sqrt{r\cdot\Exp[f]}}\leq 4 \cdot \exp\parens*{-\frac{t^2}{8c^2r\Exp[f]}}\]
\end{lemma}

\section{Missing Proofs}
\label{app:missing-proofs}

\subsection{Bounding Anti-Degree via Sparsity}

\acdproperties* 

\begin{proof}[Proof of the bound on $a_v$]

Notice that each node $u\in C\setminus N(v)$ has at least $(1-3\eps)\Delta$ common neighbors with $v$, since $|N_{C}(u)|,|N_{C}(v)|\ge (1-\eps)\Delta$, and $|C|\le (1+\eps)\Delta$;   
hence, there are at least $a_v \cdot (1-3\eps)\Delta$ edges between $N(v)$ and  $C\setminus N(v)$. On the other hand, by the definition of sparsity, at most $2\zeta_v\Delta$ edges can exit $N(v)$. Thus, $a_v\le \frac{2\zeta_v}{1-3\eps}$.
\end{proof}

\subsection{\texorpdfstring{\Cref{lem:chromaticslack}:}{Lemma~\ref{lem:chromaticslack}:} Chromatic Slack Generation}

\chromaticslack* 

\begin{proof}
Let $v\in C$. For a color $\col \in \pal'=\cup_{u\in C}(\pal_u\setminus \pal_v)$, let $Z_\col$ be the indicator random variable that is 1 iff a node $u\in C$ tries color $\col$ and no other node in $N(u)$ tries $\col$. Let $Z=\sum_{\col\in \pal'} Z_\col$. Observe that $\cs_v\ge Z-a_v$, where we conservatively assume that each node in $A_v$ contributes 1 (the most it can) to the sum $Z$. Let us bound $\Exp[Z_\col]$, for any $\col \in \pal'$, and then $\Exp[Z]$. Let $d_\col=|\{u\in C : \col \in \pal_u\}|$ be the number of nodes in $C$ that can try $\col$. The probability that $\col$ is tried by some node in $C$ is then $\pgen d_\col/(\Delta+1)$, where $\pgen =1/20$ is the probability of a node being sampled in {\slackgeneration}. Given that $\col$ is tried by a node $u$, the probability that no other node in $N(u)$ tries it is at least $(1-\pgen/(\Delta+1))^\Delta\ge 1-\pgen$; hence, we have $\Pr[Z_\col]\ge \pgen(1-\pgen)d_\col / (\Delta+1)\ge \pgen d_\col/(2\Delta)$, and
\begin{equation}\label{eq:expz}
\Exp[Z]\ge \sum_{\col \in \pal'}\frac{\pgen d_\col}{2\Delta}= \sum_{u\in C}\frac{\pgen|\pal_u\setminus \pal_v|}{2\Delta}
   = \frac{\pgen\eta_v |C|}{2\Delta}\ge \frac{\pgen\eta_v}{3}\ ,
   \end{equation}
where the second step uses the definition of  $d_\col$ and a sum rearrangement (summing over colors vs summing over vertices),  and the last one uses $|C|\ge (1-\eps)\Delta\ge (2/3)\Delta$, assuming $\eps\le 1/3$.

Observe that $Z=X-Y$, where $X$ is the number of colors tried by some node in $C$, while $Y$ is the number of colors tried and not kept by some node. By expressing $X$ as a sum of indicator variables $X_\col$, for colors $\col$, we obtain, according to the discussion above , that $\Exp[X]=\frac{\pgen\eta_v |C|}{\Delta+1}\le 2\Exp[Z]$. We also have $\Exp[Y]\le \Exp[X]\le 2\Exp[Z]$. Next, let us show that $X$ is a $1$-Lipschitz and $1$-certifiable function of the sampling outcomes and colors tried by nodes in $C$.  Indeed, changing a single such color or sampling outcome can only change $X$ by 1, and if $X\ge s$, it is enough to provide a set of $s$ colors in $\pal'$ tried by nodes in $C$. On the other hand, $Y$ is a $1$-Lipschitz and $2$-certifiable function of the colors tried and sampling outcomes of nodes in $C\cup N(C)$: changing one such value can only affect the contribution of a single color in $Y$, and in order to certify $Y\ge s$, we can provide, for $s$ colors in $\pal'$, the node that tried that color, and a neighbor of that node that also tried that color. Let $t=\Exp[Z]/3-30\sqrt{4\cdot\Exp[Z]}$. Let us apply Lemma~\ref{lem:talagrand} with $t$ and $c=1,r=2$, for each of $X$ and $Y$ (also recall that $\Exp[X],\Exp[Y]\le 2\Exp[Z]$): we get, using (\ref{eq:expz}) and the definition of $t$, 
\[
\Pr[|X-\Exp[X]|>\Exp[Z]/3\text{ and }|Y-\Exp[Y]|>\Exp[Z]/3]\le 8\exp\left(-\frac{t^2}{32\Exp[Z]}\right)\le 8\exp(-\Omega(\eta_v))\ .
\]
 Thus, w.p.\ $1-e^{-\Omega(\eta_v)}$, $\max(|X-\Exp[X]|,|Y-\Exp[Y]|)\le \Exp[Z]/3$, and hence 
\[
Z=X-Y\ge \Exp[X]-\Exp[Y]-2\Exp[Z]/3=\Exp[Z]/3\ge \pgen\eta_v/9\ .
\]
Recalling that $\cs_v\ge Z-a_v$, we obtain the proof of the lower-bound on $\cs_v$. 

On the other hand, it is easy to see that $\cs_v\le X$. Recall that $\Exp[X]=\pgen\eta_v|C|/(\Delta+1)\le (1+\eps)\pgen\eta_v$, and the concentration bound above implies that with probability $1-e^{-\Omega(\eta_v)}$, $\cs_v\le X\le \Exp[X]+\Exp[Z]/3\le 2\pgen\eta_v$, which completes the proof of the first part of the lemma. 

If $\eta_v=O(\log n)$, then $\Exp[X]=O(\log n)$, and we can apply Lemma~\ref{lem:talagrand} with $t=c\log n$, for an appropriate constant $c>0$, to show that $\cs_v\le X=O(\log n)$, w.h.p.
\end{proof}

\end{document}